\documentclass{article}
\usepackage[utf8]{inputenc}
\usepackage{algorithm}
\usepackage{algorithmic}
\usepackage{natbib}
\let\cite\citep
\usepackage{amsmath}
\usepackage{amsfonts} 
\usepackage{amsthm}
\usepackage{romannum}
\usepackage{graphicx}

\usepackage{color}
\usepackage{soul}

\usepackage{fullpage}
\newcommand{\E}{\mathbb{E}}
\newcommand{\R}{\mathbb{R}}

\newcommand{\calD}{\mathcal{D}}

\newcommand{\calN}{\mathcal{N}}

\newcommand{\calT}{\mathcal{T}}

\newcommand{\calV}{\mathcal{V}}

\newcommand{\rmB}{\mathrm{B}}
\newcommand{\rmH}{\mathrm{H}}
\newcommand{\rmL}{\mathrm{L}}
\newcommand{\rmS}{\mathrm{S}}

\newcommand{\IP}{\mathrm{IP}}
\newcommand{\IR}{\mathrm{IR}}
\newcommand{\WBB}{\mathrm{WBB}}

\newcommand{\expost}{\textit{ex post}}
\newcommand{\exante}{\textit{ex ante}}
\newcommand{\interim}{\textit{interim}}

\newtheorem{definition}{Definition}
\newtheorem{lemma}{Lemma}

\newtheorem{theorem}{Theorem}
\newtheorem{example}{Example}

\DeclareMathOperator*{\argmin}{argmin}
\DeclareMathOperator*{\argmax}{argmax}

\newcommand{\SWdone}[1]{}
\newcommand{\osogamidone}[1]{}
\newcommand{\ESdone}[1]{}

\title{Mechanism Learning for Trading Networks
}
\author{
  Takayuki Osogami\footnote{IBM Research - Tokyo.  {\tt osogami@jp.ibm.com}}
  \ \
  Segev Wasserkrug\footnote{IBM Research - Haifa.  {\tt segevw@il.ibm.com}}
  \ \
  Elisheva S. Shamash\footnote{Technion - Israel Institute of Technology.  This work was performed at IBM Research.  {\tt elisheva.shamash@gmail.com}}
}
\date{}

\begin{document}
\pagenumbering{arabic}

\maketitle

\begin{abstract}
    We study the problem of designing mechanisms for trading networks that satisfy four desired properties: dominant-strategy incentive compatibility, efficiency, weak budget balance (WBB), and individual rationality (IR).  Although there exist mechanisms that simultaneously satisfy these properties \expost{} for combinatorial auctions,  we prove the impossibility that such mechanisms do not exist for a broad class of trading networks.  We thus propose approaches for computing and learning the mechanisms that satisfy the four properties, in a Bayesian setting, where WBB and IR, respectively, are relaxed to \exante{} and \interim{}.  For computational and sample efficiency, we introduce several techniques, including game theoretical analysis to reduce the input feature space.  We empirically demonstrate that the proposed approaches successfully find the mechanisms with the four properties for those trading networks where the impossibility holds \expost{}.

\end{abstract}

\section{Introduction}

Trading networks, where firms trade via bi-lateral contracts, are becoming increasingly ubiquitous and essential to the world's economy in vital areas such as transportation and supply chain.
The purpose of such a trading network is both to benefit its participants and to contribute to the external markets.  As such, a major goal of a trading network is \textit{efficiency} in the sense of maximizing the value (social welfare) that the trading network creates \cite{hatfield2013stability}.


The prior work on trading networks primarily investigates solution concepts such as stability and competitive equilibrium \cite{hatfield2013stability,CandoganFullVersionEquilibrium}, which can be shown to imply efficiency.  To compute these solutions for a given trading network, however, one would require the full information about the value of each subset of trades to each firm (or the type of the firm), which is typically private information and needs to be truthfully revealed by the firms.  As we will show, however, the firms typically have incentive to be untruthful.


Our goal in this work is to provide a mechanism for a given trading network such that the firms in the trading network have the incentive to be truthful, which in turn leads to efficiency.  For this purpose, we follow the standard practice in mechanism design of introducing the payment to or from an Independent Party (IP) in the trading network and assume that there is a prior distribution over the types of the firms (i.e., we consider a Bayesian setting).  More specifically, in this setting, we seek to compute (or learn) the mechanism that satisfies weak budget balance (WBB; which ensures non-negative utility of IP) and individual rationality (IR; which ensures non-negative utility of firms) in addition to efficiency and incentive compatibility (which promotes truthfulness).  These four properties are standard in the literature of mechanism design \cite{parkes01}, and, for combinatorial auctions, the VCG mechanism with the Clarke pivot rule is known to satisfy all of the four properties \expost{}, i.e.\ surely for any types, \cite{IntroMechDesign}.

Notice, however, that trading networks are fundamentally more complex than combinatorial auctions.  For example, a trade is a transfer of a good from a seller to a buyer, and the seller who has negative valuation on the trade is compensated by the payment from the buyer.  We will show that, in this more complex settings, it is impossible to simultaneously satisfy all of the four properties \expost{} for a broad class of trading networks.  



We thus relax the WBB and IR properties and seek to compute a mechanism that respectively satisfies these properties \exante{} and \interim{} (i.e., in expectation with respect to the prior distribution over the types) rather than \expost{}.  To this end, we restrict ourselves to the class of Groves mechanisms to ensure incentive compatibility and efficiency and compute the pivot rule of a Groves mechanism via a linear program (LP) that encodes WBB and IR as constraints.

While this LP successfully computes a mechanism that simultaneously satisfies the four properties, it involves two fundamental challenges.  First, the mechanism designer needs the exact knowledge of the prior distribution of types.  Second, the linear program becomes intractable as the number of trades, players, or their types increases.

We thus further provide an approach of learning a mechanism from the sample of the prior distribution of types to mitigate the shortcomings of the computational approach.
With the proposed learning approach, the mechanism designer only requires the access to the sample instead of the precise knowledge of the prior distribution.  In addition, we propose several techniques to reduce the computational complexity in the proposed mechanism learning.  In particular, we manually design special Groves mechanisms, which achieve some but not all of the four properties \expost{}, and demonstrate how the dimension of the feature vector of the types may be reduced based on the knowledge of those special Groves mechanisms.

\paragraph{Contributions}  Our contributions are fourfold.  First, we prove a fundamental theorem that holds for any mechanism for trading networks with payment to/from IP (Theorem~\ref{long:thrm:nopayment}), which shows that we may ignore the payment between firms without loss of generality.  This greatly reduces the set of possible mechanisms that needs to be considered, thereby greatly simplifying the subsequent analysis and proposed algorithms.  Second, we prove the impossibility of simultaneously satisfying the four properties \expost{} in trading networks (Theorem~\ref{long:thrm:impos}), which is analogous to the Myerson–Satterthwaite theorem~\cite{MyeSat83}.  Unlike the Myerson–Satterthwaite theorem, however, our theorem does not require the assumption that the prior distribution on types has absolutely continuous density.
Third, we provide the approaches of automated mechanism design \cite{AMD,AMDCP03} for trading networks, by formulating an optimization problem and its learning variant whose solutions give the mechanisms that achieve the four properties, two in \expost{} and two in \exante{}/ \interim{}, and propose techniques to reduce the computational complexity in computing or learning the solution to these optimization problems (see Sections~\ref{long:sec:compute}-\ref{long:sec:learn}).  These include the technique of reducing the dimensionality of the feature vector based on the manually designed special Groves mechanisms.  Finally, we provide empirical support on the effectiveness of the proposed computational and learning approaches to mechanism design for trading networks (see Section~\ref{long:sec:exp}).





\paragraph{Related Work}

Trading networks and their efficiency have been extensively studied \cite{hatfield2013stability,CandoganFullVersionEquilibrium,ostrovsky2008stability,hatfield2012matching,hatfield2015full}), but no AMD approaches are known to simultaneously achieve the four properties in trading networks.  In fact, there are few studies on incentive compatibility in trading networks with an exception of \citet{schlegel2022structure}, who establishes incentive compatibility but only for those firms who are buyers (or sellers) in all trades.

From methodological perspectives, the prior work most related to ours is some of the AMD approaches for combinatorial auctions, where constrained optimization problems are formulated and approximated with the sample from the prior distribution of types.  However, in \citet{pmlr-v97-duetting19a,rahme2021auction}, DSIC is encoded as a constraint of the optimization problem and is not necessarily guaranteed due to the sample approximation unless the dataset covers the full support of the prior distribution, while we always ensure DSIC via the Groves mechanism.
On the other hand, \citet{MJG18,tacchetti2022learning} learn the mechanism (specifically, the rule of redistributing payment from IP to players) that minimizes the expected revenue of IP, while ensuring the four properties via the VCG mechanism.  In trading networks, however, we can guarantee only two properties via the Groves mechanism, and we encode the other two properties as constraints of our optimization problem. 

\citet{alon2021incomplete} also study an approach of computing (no learning) a mechanism for a principal-agent model as a solution to LP, while ensuring some of the properties via the VCG mechanism.  Earlier work along this line of computing a mechanism within a restricted class of mechanisms includes \citet{AMDRevenueMaximizingAAAI2004}, who maximize expected revenue in combinatorial auctions.  Sample approximation of this method is studied by \citet{AMDRevenueMaximizing2005}.

\section{Trading networks}
\label{long:sec:without}

Following \citet{hatfield2013stability}, we model a trading network by a tuple
$(\calN,\Omega,v)$, where $\calN$ is a set of players
(or firms), $\Omega$ is a set of bi-lateral trades, and 
$v\equiv (v_i)_{i\in\calN}$, where each $v_i:2^\Omega\to\R$ is
the type (valuation) of player $i$.
That is, $v_i(\Phi)$
represents the (possibly negative) value of $\Phi\subseteq\Omega$ for the player with type $v_i$.
Without loss of
generality, we assume that the value of no trade is zero for any
player (i.e., $v_i(\emptyset)=0, \forall i\in\calN$).
Each trade is associated with a seller
and a buyer in $\calN$, and the buyer makes non-negative payment to the seller.  Let $p(\omega)$ be the payment associated with $\omega\in\Omega$.
We refer to a pair
$(\omega,p(\omega))$ as a contract.  A pair of $\Phi$ and $p$ determines a set of contracts\footnote{$p$ may need to specify the payment for $\omega\in\Omega\setminus\Phi$ to discuss some of the properties.}, $\{(\omega,p(\omega)):
\omega\in\Phi\}$.

\SWdone{"Should we perhaps put this last part "but 
$p$ may need to specify the payment for any $\omega\in\Omega$ to discuss some of the properties" in a footnote?"}

\SWdone{I think it makes sense to explicitly say that we enhance the definition in \citet{hatfield2013stability} to include uncertainty on types, whereas \cite{hatfield2013stability} assumes the types are common knowledge. This is an extension of the model to incorporate private information.} 

One's goal with a trading network is to determine a set of trades $\Phi\subseteq\Omega$ to be conducted together with the payment associated with each trade in a way that certain properties are satisfied.  A particularly important metric is the total valuations associated with the trades to be conducted (i.e., $\sum_{i\in\calN} v_i(\Phi)$), since it is the value that the trading network produces to the external market(s).  In particular, we say that the set of trades $\Phi^\star$ is efficient for a trading network $\calT(v)=(\calN,\Omega,v)$ if
\begin{align}
  \sum_{i\in\calN} 
  v_i(\Phi^\star)
  \ge 
  \sum_{i\in\calN}
  v_i(\Phi),
  \forall \Phi\in 2^\Omega.
  \label{long:eq:Efficiency_equivalent}
\end{align}

The efficiency may also be represented in terms of the utilities of players.  Specifically, when the set of contracts $\{(\omega,p(\omega)): \omega\in\Phi\}$ is conducted, the player $i\in\calN$ gets the following (quasi-linear) utility\footnote{In the literature of trading networks (e.g., \citet{hatfield2013stability}), a valuation function is often denoted by $u_i$ and a utility by $U_i$, but we use the notations from mechanism design.  Also, we reserve $u_i$ for the utility defined in Section~\ref{long:sec:with} and use $\tilde u_i$ here.}:
\begin{align}
  \tilde u_i((\Phi,p);\calT(v))
  & = v_i(\Phi)
  + \sum_{\omega\in \Phi_{i \to}} p(\omega)
  - \sum_{\omega\in \Phi_{\to i}} p(\omega),
  \label{long:eq:utility_without}
\end{align}
where $\Phi_{i\to}$ is the subset of $\Phi$ where $i$ is the seller, and $\Phi_{\to i}$ is the one where $i$ is the buyer.
Now, since we have
\begin{align}
\bigcup_{i\in\calN} \Phi_{i \to}
= \bigcup_{i\in\calN} \Phi_{\to i}
= \Phi
\label{long:eq:Phi}
\end{align}
for any $\Phi\subseteq\Omega$,
the total payment from sellers equals the total payment to buyers:
\begin{align}
    \sum_{i\in\calN} \sum_{\omega\in \Phi_{i \to}} p(\omega)
    = \sum_{i\in\calN} \sum_{\omega\in \Phi_{\to i}} p(\omega),
    \label{long:eq:total_payment}
\end{align}
which together with \eqref{long:eq:utility_without} implies
\begin{align}
    \sum_{i\in\calN} \tilde u_i((\Phi,p);\calT(v))
    = \sum_{i\in\calN} v_i(\Phi)
\end{align}
for any $p$.  Therefore, for any $p$, \eqref{long:eq:Efficiency_equivalent} is equivalent to
\begin{align}
  \sum_{i\in\calN} \tilde u_i((\Phi^\star,p); \calT(v))
  \ge 
  \sum_{i\in\calN} \tilde u_i((\Phi,p); \calT(v)),
  \forall \Phi\in 2^\Omega.
  \label{long:eq:Efficiency_without}
\end{align}

The prior work has developed several algorithms for finding the set of contracts that achieve the properties that are stronger than efficiency (e.g., competitive equilibrium and stability) under certain conditions (e.g., full substitutability \cite{FullSubstitutability}).  However, they all rely on the knowledge of the types $v$.

In practice, however, the types are private information and need to be revealed by the players.  Then, the players can have the incentive to be untruthful, which in turn leads to inefficiency, as the following example shows:
\begin{example}
  Consider a trading network with a single potential trade,
  $\Omega=\{\omega\}$, between two players, $\calN=\{\rmS,\rmB\}$.
  With $\omega$, the seller (S) incurs a production cost $C_\rmS$, and
  the buyer (B) gets a profit $1-C_\rmB$ (retail price minus handling
  cost): i.e., $v_\rmS(\Phi)=-C_\rmS\,I\{\Phi=\Omega\}$ and
  $v_\rmB(\Phi)=(1-C_\rmB)\,I\{\Phi=\Omega\}$, where $I\{\cdot\}$ is
  the indicator function.  Consider the payment,
  $p(\omega)=\frac{1+C_\rmS-C_\rmB}{2}$, that equally shares the net
  profit between S and B.  Given the types (i.e., $C_S$ and $C_B$), we can
  achieve efficiency by letting $\Phi=\Omega$ if $C_\rmS+C_\rmB<1$ and
  $\Phi=\emptyset$ otherwise.  When $C_\rmS+C_\rmB<1$, the utility of
  each player is $\frac{1-C_\rmS-C_\rmB}{2}$.  However, then each
  player $i$ has the incentive to untruthfully declare slightly
  higher cost $C_i+\varepsilon$ to get slightly higher utility
  $\frac{1-C_\rmS-C_\rmB+\varepsilon}{2}$.  If the declared cost is
  too high, the trade is not conducted.  Since
  $1-C_\rmS-C_\rmB=v_\rmS(\Omega)+v_\rmB(\Omega)>v_\rmS(\emptyset)+v_\rmB(\emptyset)=0$,
  untruthfulness indeed leads to inefficiency.
  \label{long:ex:nontruthful}
\end{example}

\section{Introducing Payment to/from IP}
\label{long:sec:with}

We consider the trading networks where the types are private information of the respective players and study the mechanisms that promote truthfulness and hence lead to efficiency.  To facilitate truthfulness, the mechanism may require each player $i$ to make (possibly negative) payment $\tau_i\in\R$ to an independent party (IP).  A trading network with IP is denoted by $\calT^+(v)=(\calN^+,\Omega,v)$, where $\calN^+\equiv\calN\cup\{\IP\}$ denotes the set of all players and IP.  Also, let $\calV_i$ be the space of types for each $i\in\calN$ and let $\calV\equiv\times_{i\in\calN}\calV_i$. Moreover, let $\calT^+(\calV)\equiv(\calN,\Omega,\calV)\equiv\{\calT^+(v): v\in\calV\}$ be the set of trading networks with IP under $\calV$.  We will simply refer to $\calT^+(v)$ or $\calT^+(\calV)$ as a trading network.

We study the Bayesian setting where there exists a prior distribution $q$ over $\calV$ such that the players have types $v$ with probability $q(v)$.  Throughout, we assume that the true type $v_i$ of each player $i\in\calN$ is known by that player.  We will, however, make varying assumptions on other knowledge about the types, which will be made explicit in each of the results in the following.


\osogamidone{I am not sure if we exclusively study Bayesian setting.  Why DSIC instead of BNIC?  Why ex post?  Bayesian with uncertainties in the prior?  Ex-ante properties are useless if prior is wrong.}

\subsection{Direct Mechanisms}

We study a direct mechanism, where players and IP act according to the following protocol:
\begin{enumerate}
    \item Each player $i$ declares a (possibly untruthful) type $\hat v_i\in\calV_i$
    \item IP determines the set of contracts, $\{(\omega,\pi(\omega;\hat v)): \omega\in\phi(\hat v)\}$, and (possibly negative) payment, $\tau_i(\hat v)$, from each player $i$ to IP
\end{enumerate}
Hence, the direct mechanism of a trading network $\calT(\calV)$ is specified by an
outcome rule $(\phi,\tau,\pi)$, where $\phi: \calV\to2^\Omega$ is the
allocation rule that maps declared types, $\hat
v\equiv(\hat v_i)_{i\in\calN}$, to a set of trades $\phi(\hat
v)$ to be conducted; $\tau$ determines the rule of payment to IP,
$\tau_i\in\calV\to\R$, for each $i\in\calN$; $\pi:2^\Omega\times\calV\to\R$ is the payment rule that determines the payment associated with each trade, depending on the declared types.  In the following, we
refer to a direct mechanism simply as a mechanism.

When $\hat v$ is declared by the players in the trading network $\calT^+=(\calN^+,\Omega,v)$, each player $i\in\calN$ gets the following utility under the mechanism $(\phi,\tau,\pi)$:
\begin{align}
    u_i(\hat v; (\phi,\tau,\pi), \calT^+(v))
    & = v_i(\phi(\hat v))
    + \sum_{\omega\in \phi(\hat v)_{i \to}} \pi(\omega;\hat v)
    - \sum_{\omega\in \phi(\hat v)_{\to i}} \pi(\omega;\hat v)
    - \tau_i(\hat v).
    \label{long:eq:original_utility}
\end{align}
We denote the net-payment to IP (or utility of IP) by
\begin{align}
    u_{\IP}(\hat v; (\phi,\tau,\pi),\calT^+(v))
    & = \sum_{i\in\calN} \tau_i(\hat v).
    \label{long:eq:original_utility_IP}
\end{align}

\subsection{Desirable Properties}
\label{long:sec:property}


We consider the four desirable properties of a mechanism that is standard in mechanism design \cite{parkes01}: Dominant Strategy Incentive Compatibility (DSIC), Efficiency, Weak Budget Balance (WBB), and Individual Rationality (IR). These properties are known to be simultaneously achievable in combinatorial auctions, but we will see that such positive results do not carry over to trading networks.  We formally define each property in the following, but note that, under the Bayesian setting, we can discuss both \expost{} properties (which hold surely for any $v\in\calV$) and \exante{} or \interim{} properties (which hold in expectation with respect to the prior distribution $q$ over $\calV$).

We will formally define each of these properties in the following (see \citet{parkes01} for more details).  Note that, under the Bayesian setting, we can discuss either \expost{} properties (which hold surely for any $v\in\calV$) or \exante{}/\interim{} properties (which hold in expectation with respect to the prior distribution $q$ over $\calV$).


DSIC is an \expost{} property and ensures that the best strategy of each player is truthfully revealing its type regardless of the strategies of the other players.  This gives a clear course of actions.  Formally, we say that a mechanism $(\phi,\tau,\pi)$ for a trading network $\calT^+(\calV)$ is DSIC if the profile of truthful strategies form a dominant-strategy equilibrium:
\begin{align}
  u_i((v_i,v_{-i});(\phi,\tau,\pi),\calT^+(v))
  & \ge u_i((v_i',v_{-i});(\phi,\tau,\pi),\calT^+(v)), \forall (v,v_i')\in\calV\times\calV_i
  \label{long:eq:DSIC}
\end{align}
where $v_{-i}$ is the strategy profile of players except $i$.  We will discuss the corresponding \exante{} property of Bayesian Nash Incentive-Compatibility (BNIC) only in relation to the prior work, but BNIC only ensures that a truthful player can maximize its \emph{expected} utility under the additional condition that other players are truthful.



Efficiency is also an \expost{} property and ensures that the trades are allocated in a way that they maximize total value.  Following the definition without IP in \eqref{long:eq:Efficiency_equivalent}, we say that a mechanism $(\phi,\tau,\pi)$ for $\calT^+(\calV)$ is Efficient if
\begin{align}
  \sum_{i\in\calN} 
  v_i(\phi(v))
  \ge 
  \sum_{i\in\calN}
  v_i(\Phi),
  \forall (v,\Phi)\in \calV\times 2^\Omega.
  \label{long:eq:Efficiency}
\end{align}
We will not discuss the corresponding \exante{} property.

WBB ensures nonnegative utility of IP.  We say that a mechanism $(\phi,\tau,\pi)$ for $\calT^+(\calV)$ is \expost{} WBB if
\begin{align}
  u_\IP(v;(\phi,\tau,\pi),\calT^+(v))
  \ge 0, \forall v\in\calV
  \label{long:eq:expostWBB}
\end{align}
and \exante{} WBB if
\begin{align}
  \E\left[u_\IP(v;(\phi,\tau,\pi),\calT^+(v))\right]
  \ge 0,
  \label{long:eq:exanteWBB}
\end{align}
where the expectation is with respect to the prior distribution of $v$.  That is, under the additional condition that the players are truthful (which is guaranteed with DSIC), the utility of IP is non-negative surely under \expost{} WBB and in expectation under \exante{} WBB.

IR ensures that every player gets non-negative utility.  IR can also be \expost{} or \exante{}.  However, when players know their own types, each player should require non-negative expected utility \emph{given its own type}, and such a property is called \interim{} IR.  Formally, we say that a mechanism $(\phi,\tau,\pi)$ for $\calT^+(\calV)$ is \expost{} IR if
\begin{align}
  u_i(v;(\phi,\tau,\pi),\calT^+(v))
  \ge 0, \forall v\in\calV
  \label{long:eq:expostIR}
\end{align}
for each $i\in\calN$ and \interim{} IR if
\begin{align}
  \E\left[ u_i(v;(\phi,\tau,\pi),\calT^+(v)) \mid v_i \right]
  \ge 0, \forall v_i\in\calV_i
  \label{long:eq:interimIR}
\end{align}
for each $i\in\calN$.  That is, under the additional condition that the players are truthful (guaranteed with DSIC), a truthful player gets non-negative utility surely under \expost{} IR and in expectation under \interim{} IR.

Note that \expost{} properties are more desirable than the corresponding \exante{} properties.  For example, maximizing expected utility may not be the objective of risk-sensitive players.  Also, the optimality of the truthful strategy under BNIC relies on the truthfulness of the other players, which is not required under DSIC.

\subsection{No Payment between Players}
\label{long:sec:with:payment}

We claim that, as far as the utilities of the forms in \eqref{long:eq:original_utility}-\eqref{long:eq:original_utility_IP} are concerned, we do not lose generality by ignoring the payment among players (i.e., $\pi(\omega;\hat v)=0, \forall (\omega,\hat v)$).  Formally, 
\begin{lemma}
  For any mechanism $(\phi,\tau,\pi)$ and a payment rule $\pi'$ for a trading network
  $\calT^+(\calV)=(\calN^+,\Omega,\calV)$, there exists a rule of payment to IP $\tau'$
  such that $u_i(\hat
  v;(\phi,\tau,\pi),\calT^+(v))=u_i(\hat v;(\phi,\tau',\pi'),\calT^+(v))$ for any $(\hat
  v,v,i)\in\calV^2\times\calN^+$.
  \label{long:lemma:nopayment}
\end{lemma}
\begin{proof}
Let
\begin{align}
    \tau_i'(\hat v)
    & =
    \tau_i(\hat v)
    - \sum_{\omega\in \phi(\hat v)_{i \to}} \pi(\omega;\hat v)
    + \sum_{\omega\in \phi(\hat v)_{\to i}} \pi(\omega;\hat v)
    + \sum_{\omega\in \phi(\hat v)_{i \to}} \pi'(\omega;\hat v)
    - \sum_{\omega\in \phi(\hat v)_{\to i}} \pi'(\omega;\hat v).
    \label{long:eq:payment_conversion}
\end{align}
Then, by \eqref{long:eq:original_utility}, we
have $u_i(\hat v; (\phi,\tau,\pi), \calT^+(v))=u_i(\hat v; (\phi,\tau',\pi'),\calT^+(v))$ for
any $(v,i)\in\calV\times\calN$.  
Also, \eqref{long:eq:total_payment} implies $u_{\IP}(\hat v; (\phi,\tau,\pi),\calT^+(v))=u_{\IP}(\hat v; (\phi,\tau',\pi'),\calT^+(v))$ for any $v\in\calV$ by the definition of $u_{\IP}$ in \eqref{long:eq:original_utility_IP}.
\end{proof}

The following theorem implies that we only need to consider the mechanisms without payment between players to achieve the four properties:
\begin{theorem}
  For any mechanism $(\phi,\tau,\pi)$ and a payment rule $\pi'$ for
  $\calT^+(\calV)=(\calN^+,\Omega,\calV)$, there exists a rule of payment to IP, $\tau'$
  such that each of DSIC, Efficiency, WBB
  (either \exante{} or \expost{}), and IR (either \exante{} or \interim{}) is
  satisfied with $(\phi,\tau,\pi)$ iff it is satisfied
  with $(\phi,\tau',\pi')$.
  \label{long:thrm:nopayment}
\end{theorem}
\begin{proof}
  By Lemma~\ref{long:lemma:nopayment}, for any pair of $(\phi,\tau,\pi)$ and $\pi'$, there
  exists $\tau'$ such that utilities of the players and the IP
  under $(\phi,\tau,\pi)$ are the same as those under $(\phi,\tau',\pi')$ for any $v\in\calV$.  Hence, it suffices to show
  that each of the properties depend only on the utilities of the
  players and the IP.  For example, DSIC \eqref{long:eq:DSIC} depends only
  on the utilities; hence, if $(\phi,\tau,\pi)$ is DSIC, i.e.,
  \begin{align}
    u_i((v_i,v_{-i});(\phi,\tau,\pi),\calT^+(v))
    \ge u_i((v_i',v_{-i});(\phi,\tau,\pi),\calT^+(v)), \forall(v,v_i')\in\calV\times\calV_i,
  \end{align}
  then, since $u_i(v;(\phi,\tau,\pi),\calT^+(v))=u_i(v;(\phi,\tau',\pi'),\calT^+(v))$ for any $v$, we must have
  \begin{align}
    u_i((v_i,v_{-i});(\phi,\tau',\pi'),\calT^+(v))
    \ge u_i((v_i',v_{-i});(\phi,\tau',\pi'),\calT^+(v)), \forall(v,v_i')\in\calV\times\calV_i,
  \end{align}
  i.e., $(\phi,\tau')$ is DSIC.  The converse also holds.

  Since it is straightforward to see that \expost{} WBB \eqref{long:eq:expostWBB}, \exante{} WBB \eqref{long:eq:exanteWBB}, \expost{} IR \eqref{long:eq:expostIR}, and \interim{} IR \eqref{long:eq:interimIR} also depend only on the utilities of the players and the IP, it remains to show that Efficiency depends only on the utilities.  Now, observe that \eqref{long:eq:Efficiency} is equivalent to
\begin{align}
    \sum_{i\in\calN^+} u_i(v;(\phi,\tau,\pi),\calT^+(v))
    \ge 
    \sum_{i\in\calN^+} u_i(v;(\phi',\tau,\pi),\calT^+(v)),
    \forall v\in \calV, \forall\phi',
\end{align}
which depends only on the utilities of the players and the IP, which establishes the theorem.
\end{proof}


Theorem~\ref{long:thrm:nopayment} implies that we only need to consider mechanisms with no payment between players as long as the desirable properties are the ones that depend only on the utility.  Namely, for any set of properties that depend on the utility, there exists a mechanism with no payment between players that achieves those properties if and only if there exists a mechanism (with an arbitrary payment rule) that achieves those properties.  Hence, we will only consider designing a mechanism $(\phi,\tau,\pi_0)$ without payment between players and denote such a mechanism by $(\phi,\tau)$.  Intuitively, if there is a payment $p(\omega)$ from a buyer to a seller, we can let the buyer pay $p(\omega)$ to IP and let IP pay $p(\omega)$ to the seller without changing the net-payment of the two players and IP.

\section{Impossibility}
\label{long:sec:imposs}

Here, we show that no mechanism $(\phi,\tau)$ can achieve all of the four properties \expost{} for trading networks with a single potential trade except for trivial cases.  Theorem~\ref{long:thrm:nopayment} then implies that no mechanisms $(\phi,\tau,\pi)$, with any payment rule $\pi$, can simultaneously achieve the four properties except for trivial cases.  Formally,
\begin{definition}
 We say that a trading network $(\calN^+,\{\omega\},\calV)$ with a single trade between two players, $|\calN|=2$, has non-trivial $\calV$ if it satisfies both of the following conditions:
  \romannum{1}) there exists $v\in\calV$ such that $\sum_{j\in\calN} v_j(\{\omega\})<0$, and
  \romannum{2}) for any $v_i\in\calV_i$, there exists $v_{-i}\in\calV_{-i}$ such that $\sum_{j\in\calN} v_j(\{\omega\})>0$, where $\{i,-i\}\in\calN$.
  \label{long:def:trivial}
\end{definition}
In other words, $\calV$ is said to be trivial if at least one of the two conditions in Definition~\ref{long:def:trivial} are violated.
When condition \romannum{1}) is violated, the trade should always be conducted (i.e., $\phi(\hat v)=\{\omega\}, \forall \hat v\in\calV$) for Efficiency to hold, and all of the four properties are trivially satisfied with a constant payment rule\footnote{The simple example in \cite{SandhomOthman2009PervasiveMS}, where impossibility is shown not to hold, violates condition \romannum{1}.}.  Specifically, consider a seller (S) and a buyer (B), and let $\underline{v}_i\equiv\min_{v_i\in\calV_i} v_i(\{\omega\})$ be the minimum value of the trade for each $i\in\calN=\{\rmS,\rmB\}$.  Since condition \romannum{1}) is violated, we have $\underline{v}_\rmS+\underline{v}_\rmB\ge 0$.  Without loss of generality, let $\underline{v}_\rmS \le \underline{v}_\rmB$.  Then, with the constant payment rule of $\tau'_\rmS(\hat v)=-\underline{v}_\rmB$ and $\tau'_\rmB(\hat v)=\underline{v}_\rmB$, it is straightforward to verify that the four conditions are satisfied.

When condition \romannum{2}) is violated, there exists a player $i$ who has the type $v_i$ that makes it impossible to make the social welfare, $v_i(\{\omega\})+v_{-i}(\{\omega\})$, positive no matter what types $v_{-i}$ that the other player has.  Such player $i$ should not participate in the trading network from the perspective of social welfare, even if all of the four properties may be achieved with some mechanisms.

For other non-trivial cases, no mechanisms achieve all of the four properties.  Formally,
\begin{theorem}
  For any trading network $(\calN^+,\{\omega\},\calV)$ with a single trade between two players having non-trivial $\calV$, no mechanisms can achieve all of DSIC, Efficiency, \expost{} WBB, and \expost{} IR.
  \label{long:thrm:impos}
\end{theorem}
\begin{proof}
  Let $\calN=\{\rmS,\rmB\}$ be the two players.  For each $i\in\calN$, let $v_i^\rmL\equiv\argmin_{v_i\in\calV_i} v_i(\{\omega\})$ be the type that gives the lowest value on the trade, and let $v_i^\rmH\equiv\argmax_{v_i\in\calV_i} v_i(\{\omega\})$ be the one giving the highest value.  We say that player $i$ with $v_i^\rmL$ has type $X$ for $X\in\{\rmL,\rmH\}$.  We will show that we cannot simultaneously satisfy the conditions for the four properties associated with the players of those types (L and H).

  We first consider the allocation rule $\phi(X,Y)$, where $(X,Y)$ are the types declared by (S, B).  Efficiency \eqref{long:eq:Efficiency} and the non-triviality of $\calV$ imply that we must have
  \begin{align}
    \phi(X,Y) = \left\{\begin{array}{ll}
    \emptyset & \mbox{if } (X,Y)=(\rmL,\rmL)\\
    \{\omega\} & \mbox{otherwise}
    \end{array}\right.
    \label{long:eq:impos:allocation}
  \end{align}

  Next, consider the rule of payment to IP, $\tau'(X,Y)$.  DSIC \eqref{long:eq:DSIC} requires among others that
  \begin{align}
    v_\rmS^\rmH(\phi(\rmH,\rmH)) - \tau'_\rmS(\rmH,\rmH)
    & \ge v_\rmS^\rmH(\phi(\rmL,\rmH)) - \tau'_\rmS(\rmL,\rmH)
    \label{long:eq:impos:dsic1}\\
    v_\rmS^\rmL(\phi(\rmL,\rmH)) - \tau'_\rmS(\rmL,\rmH)
    & \ge v_\rmS^\rmL(\phi(\rmH,\rmH)) - \tau'_\rmS(\rmH,\rmH)\\
    v_\rmB^\rmH(\phi(\rmH,\rmH)) - \tau'_\rmB(\rmH,\rmH)
    & \ge v_\rmB^\rmH(\phi(\rmH,\rmL)) - \tau'_\rmB(\rmH,\rmL)\\
    v_\rmB^\rmL(\phi(\rmH,\rmL)) - \tau'_\rmB(\rmH,\rmL)
    & \ge v_\rmB^\rmL(\phi(\rmH,\rmH)) - \tau'_\rmB(\rmH,\rmH).
    \label{long:eq:impos:dsic4}
  \end{align}
  Since $\phi(\rmH,\rmH)=\phi(\rmL,\rmH)=\phi(\rmH,\rmL)=\{\omega\}$ by \eqref{long:eq:impos:allocation}, we can reduce \eqref{long:eq:impos:dsic1}-\eqref{long:eq:impos:dsic4} to
  \begin{align}
    \tau'_\rmS(\rmL,\rmH) & = \tau'_\rmS(\rmH,\rmH)
    \label{long:eq:impos:DSIC:1}\\
    \tau'_\rmB(\rmH,\rmL) & = \tau'_\rmB(\rmH,\rmH).
    \label{long:eq:impos:DSIC:2}
  \end{align}
  \expost{} IR \eqref{long:eq:expostIR} requires among others that
  \begin{align}
    \tau'_\rmS(\rmL,\rmH) & \le v_\rmS^\rmL(\phi(\rmL,\rmH)) = v_\rmS^\rmL(\{\omega\})
    \label{long:eq:impos:IR:1}\\
    \tau'_\rmB(\rmH,\rmL) & \le v_\rmB^\rmL(\phi(\rmH,\rmL)) = v_\rmB^\rmL(\{\omega\})
    \label{long:eq:impos:IR:2}
  \end{align}
  By \eqref{long:eq:impos:DSIC:1}-\eqref{long:eq:impos:IR:2}, we must have
  \begin{align}
    \tau'_\rmS(\rmH,\rmH) + \tau'_\rmB(\rmH,\rmH) \le v_\rmS^\rmL(\{\omega\}) + v_\rmB^\rmL(\{\omega\})
  \end{align}

  Since $v_\rmS^\rmL(\{\omega\}) + v_\rmB^\rmL(\{\omega\})<0$ by condition~\romannum{1}) of Definition~\ref{long:def:trivial}, this contradicts a requirement from \expost{} WBB \eqref{long:eq:expostWBB}:
  \begin{align}
    \tau'_\rmS(\rmH,\rmH) + \tau'_\rmB(\rmH,\rmH) \ge 0
  \end{align}
\end{proof}

\citet{MyeSat83} show similar impossibility but assumes that the prior distribution over $\calV$ has absolutely continuous density, which does not hold e.g.\ for discrete $\calV$ \cite{SandhomOthman2009PervasiveMS}.  Our impossibility theorem does not require such an assumption, but ours is with DSIC and hence weaker in that respect than Myerson-Satterthwaite's, which establishes impossibility with BNIC.

\section{Computing Mechanisms}
\label{long:sec:compute}



The impossibility suggested by Theorem~\ref{long:thrm:impos} motivates us to study the mechanisms that achieve weaker properties of \exante{} WBB and \interim{} IR in addition to DSIC and Efficiency.  These weaker properties are still meaningful in practice and are sufficient for risk-neutral players.






To guarantee DSIC and Efficiency, we rely on the Groves mechanism.  For any trading network $(\calN^+,\Omega,\calV)$, a mechanism $(\phi,\tau')$ is called a Groves mechanism if it can be represented by the use of a pivot rule, $h\equiv(h_i)_{i\in\calN}$ where $h_i: \calV_{-i}\to\R$, as follows:
\begin{align}
  \phi(v)
  & = \phi^\star(v)
  \in \argmax_{\Phi\subseteq\Omega} \sum_{i\in\calN} v_i(\Phi) \label{long:eq:groves_allocation}\\
  \tau'_i(v)
  & = h_i(v_{-i}) - \sum_{j\neq i} v_j(\phi^\star(v)) \label{long:eq:groves_payment}
\end{align}
for any $v\in \calV$.  Here, the pivot rule $h_i$ determines the payment, in addition to the second term of the right-hand side of \eqref{long:eq:groves_payment}, in a way that it depends only on the types of other players (and not the type of player $i$).  The Groves mechanism is guaranteed to satisfy Efficiency and DSIC, and we may arbitrarily choose $h$ to achieve other properties.  We refer to $h$ as a Groves mechanism when it does not cause confusion.


\SWdone{Note: Are we focusing on deterministic mechanisms, or can they be probabilistic, i.e. probabilistically chose a maximizing assignment?}
\osogamidone{They can be probabilistic, and our results hold surely (for any realization of probabilistically chosen assignment) if they are probabilistic.}

With a Groves mechanism $h$, the condition for \exante{} WBB \eqref{long:eq:exanteWBB} is reduced to
\begin{align}
  0
  & \le \E\left[\sum_{i\in\calN} \tau'_i(v)\right] \label{long:eq:groves-wbb:from}\\
  & = \E\left[
    \sum_{i\in\calN} \left( h_i(v_{-i}) - \sum_{j\neq i} v_j(\phi^\star(v)) \right)
    \right] \\
  & = \sum_{i\in\calN} \E\left[ h_i(v_{-i}) \right]
  - (|\calN| - 1) \sum_{i\in\calN} \E\left[ v_i(\phi^\star(v)) \right],
  \label{long:eq:groves-wbb}
\end{align}
and \interim{} IR \eqref{long:eq:interimIR} is reduced to
\begin{align}
  0
  & \le \E\left[ 
    v_i(\phi^\star(v)) - \tau'_i(v)
    \mid v_i
  \right] \label{long:eq:groves-ir:from}\\
  & = \E\left[ 
    v_i(\phi^\star(v)) - h_i(v_{-i}) + \sum_{j\neq i} v_j(\phi^\star(v))
    \mid v_i
    \right]\\
  & = \sum_{j\in\calN} \E\left[ v_j(\phi^\star(v)) \mid v_i \right]
  - \E\left[ h_i(v_{-i}) \mid v_i \right].
  \label{long:eq:groves-ir}
\end{align}


There may or may not exist an $h$ that satisfies both \exante{} WBB and \interim{} IR.  When it does not exist, one could seek to find the $h$ that minimizes the violation of these conditions.  However, here we focus on studying whether such an $h$ exists and finding one if it does.  Also, when there are multiple $h$ that satisfy both \exante{} WBB and \interim{} IR, we prefer the one with stronger budget balance (i.e., the total payment to IP should be as close to zero as possible).

\SWdone{Note: Maybe we should emphasize that we have formulated this as the following optimization problem, as this is one of our contributions. In general, we should probably emphasise in the text where our contributions are} 

Therefore, we proposed to find $h$ as the solution to the following optimization problem:
\begin{align}
  \min_{h}\
  & \sum_{i\in\calN} \E\left[ h_i(v_{-i}) \right] \label{long:eq:lp:obj}\\
  \mbox{s.t. }
  & \sum_{i\in\calN} \E\left[ h_i(v_{-i}) \right]
  \ge (|\calN| - 1) \sum_{i\in\calN} \E\left[ v_i(\phi^\star(v)) \right] \label{long:eq:lp:wbb}\\
  & \E\left[ h_i(v_{-i}) \mid v_i \right]
  \le \sum_{j\in\calN} \E\left[ v_j(\phi^\star(v)) \mid v_i \right], \forall v_i\in\calV_i, \forall i\in\calN \label{long:eq:lp:ir}
\end{align}
where $\E$ is the expectation with respect to the prior distribution of types.  The constraint \eqref{long:eq:lp:wbb} is on \exante{} WBB and corresponds to \eqref{long:eq:groves-wbb}.  The constraint \eqref{long:eq:lp:ir} is on \interim{} IR and corresponds to \eqref{long:eq:groves-ir}.  for a given trading network $(\calN^+,\Omega,\calV)$, the right-hand sides of \eqref{long:eq:lp:wbb}-\eqref{long:eq:lp:ir} are constants that can be computed by the use of the allocation rule $\phi^\star$ and the prior distribution of types
When $\calN$ and $\calV$ are finite, the optimization problem \eqref{long:eq:lp:obj}-\eqref{long:eq:lp:ir} is a linear program (LP), where variables are $h_i(v_{-i})$ for $v_{-i}\in\calV, i\in\calN$.  When there are $n=|\calN|$ players, and each player $i$ has $m=|\calV_i|$ types, the LP has $n \, m^{n-1}$ variables and $n\,m+1$ constraints.

\begin{example}
  Consider the trading network in Example~\ref{long:ex:nontruthful}, but we now assume that each player has one of two types: the production cost of the seller of type $X$ is $C_\rmS^X$ for $X\in\{\rmL,\rmH\}$, and the handling cost of the buyer of type $Y$ is $C_\rmB^Y$ for $Y\in\{\rmL,\rmH\}$.  Let $C_i^\rmL < C_i^\rmH$ for $i\in\{\rmS,\rmB\}$.  By Theorem~\ref{long:thrm:impos}, there is no mechanism that simultaneously achieves the four properties \expost{} when $c_\rmS^\rmH+c_\rmB^\rmH > 1$ and $c_\rmS^X+c_\rmB^Y < 1$ for $(X,Y)\in\{(\rmH,\rmL),(\rmL,\rmH),(\rmL,\rmL)\}$.  Under the Groves mechanism, the amount of payment to IP from S of type $X$ and B of type $Y$ is respectively given by
  \begin{align}
    \tau'_S(X,Y) & = h_\rmS(Y) - (1-C_\rmB^Y) \, q^\star(X,Y) \\
    \tau'_B(X,Y) & = h_\rmB(X) + C_\rmS^X \, q^\star(X,Y),
  \end{align}
  where $q^\star(X,Y)\equiv I\{\phi^\star(X,Y)=\Omega\}$ is the number of goods to be traded when the types of S and B are $X$ and $Y$ respectively.  Note that the pivot rule of each player depends on the type of the other player and can be determined by solving the optimization problem \eqref{long:eq:lp:obj}-\eqref{long:eq:lp:ir}, which is reduced to the following for this example:
  \begin{align}
    \min_{h}\ 
    & \E[h_\rmS(Y) + h_\rmB(X)] \label{long:ex:two:obj}\\
    \mbox{s.t.}\
    & \E[h_\rmS(Y) + h_\rmB(X)] \ge \E[(1-C_\rmS^X-C_\rmB^Y) \, q^\star(X,Y)] \\
    & \E[h_\rmS(Y)\mid X] \le \E[(1-C_\rmS^X-C_\rmB^Y) \, q^\star(X,Y) \mid X], \forall X\in\{\rmL,\rmH\} \\
    & \E[h_\rmB(X)\mid Y] \le \E[(1-C_\rmS^X-C_\rmB^Y) \, q^\star(X,Y) \mid Y], \forall Y\in\{\rmL,\rmH\}
    \label{long:ex:two:constraint}
  \end{align}
  where the expectation is with respect to the prior distribution of types.  This is an LP with four variables $h\equiv(h_\rmS(\rmL),h_\rmS(\rmH),h_\rmB(\rmL),h_\rmB(\rmH))$ and four constraints.
\label{long:ex:two}
\end{example}

Although the solution to the optimization problem \eqref{long:eq:lp:obj}-\eqref{long:eq:lp:ir} indeed gives non-trivial Groves mechanisms that are of practical interest, it is intractable except for small trading networks.  For example, \romannum{1}) there can be infinitely many types $v_i$, which results in infinitely many variables in the optimization problem.  Also, \romannum{2}) one may not know the exact prior distribution of types, which we need to compute the expectation in \eqref{long:eq:lp:obj}-\eqref{long:eq:lp:ir}.  In addition, \romannum{3}) it may be hard to compute the efficient allocation $\phi^\star(v)$ for a given $v\in\calV$. 


\section{Learning Mechanisms}
\label{long:sec:learn}

Among the three challenges, \romannum{1}) and \romannum{2}) are fundamental, since one cannot even represent the optimization problem with infinitely many variables or without knowing the prior distribution.  Here, we propose learning techniques to overcome these challenges.  For the computational complexity of $\phi^\star$, efficient algorithms are known under some conditions on $\calV$ \cite{CandoganFullVersionEquilibrium,IwataSubmodularFlow2005}.



\subsection{Learning Pivot Rules}
\label{long:sec:learn:pivot}

We now assume that the mechanism designer (IP) has access to the sample $\calD$ of types from the prior distribution $q$, rather than the exact knowledge on $q$.  The expectation in \eqref{long:eq:lp:obj}-\eqref{long:eq:lp:ir} can then be replaced with sample average, which solves Challenge \romannum{2}).  IP may collect such $\calD$ by running a Groves mechanism that is not necessarily \exante{} WBB (hence, IP needs investment).  Since players act truthfully under the Groves mechanism, the collected $\calD$ should follow $q$\footnote{Assuming that players (or their types) are sampled independently each time.}.
One may then learn a better Groves mechanism (which less violates WBB) and collect additional data with the new mechanism for learning even better mechanisms (with the hope of eventually achieving \exante{} WBB).

Recall that the variables in \eqref{long:eq:lp:obj}-\eqref{long:eq:lp:ir} correspond to the output values of the pivot rule, $h_i(v_{-i})$ for $v_{-i}\in\calV_{-i}$ and $i\in\calN$, which constitute the codomain of the functions $h=(h_i)_{i\in\calN}$.  To deal with infinitely many variables that stem from the functions with infinite codomain, we approximate those functions with machine learning models, $h^\theta\equiv(h_i^{\theta_i})_{i\in\calN}$ where $h_i^{\theta_i}: \calV_{-i}\to\R$ with parameter $\theta_i$ for each $i$.

Such approaches of machine learning have been successfully applied in the prior work of mechanism learning.
Unfortunately, in our case, it is not sufficient to replace the expectation with sample average and $h$ with $h^\theta$.  In particular, the constraint \eqref{long:eq:lp:ir} is problematic.  Since the set $\{\tilde v\in\calD\mid\tilde v_i=v_i\}$ is often small and can be empty for some $v_i\in\calV_i$, the sample average is unreliable or cannot be obtained with the empty set.  A solution to this problem is to learn, for each $i\in\calN$, a regressor (e.g., Gaussian process regressor) that maps $v_i$ to $\sum_{j\in\calN} \E[v_j(\phi^\star(v))\mid v_i]$ for $v_i\in\calV_i$, where the training data is $\hat D_i \equiv \{(\tilde v_i, \sum_{j\in\calN} \tilde v_j(\phi^\star(\tilde v)))\}_{\tilde v\in\calD}$.  Let $\hat g_i$ be such a regressor trained with $\hat D_i$ for $i\in\calN$.

We can now reduce \eqref{long:eq:lp:obj}-\eqref{long:eq:lp:ir} to the following constrained non-linear optimization and learn a locally optimal $\theta$ with the augmented Lagrangian method:
\begin{align}
  \min_{\theta}\
  & \sum_{\tilde v\in\calD} \sum_{i\in\calN} h_i^{\theta_i}(\tilde v_{-i}) \label{long:eq:learn:obj}\\
  \mbox{s.t. }
  & \sum_{\tilde v\in\calD} \sum_{i\in\calN} h_i^{\theta_i}(\tilde v_{-i}) 
  \ge (|\calN| - 1) \sum_{\tilde v\in\calD} \sum_{i\in\calN} \tilde v_i(\phi^\star(\tilde v)) \label{long:eq:learn:wbb}\\
  & \frac{1}{|\{v\in\calD\mid v_i=\tilde v_i\}|} \sum_{v\in\calD\mid v_i=\tilde v_i} h_i^{\theta_i}(\tilde v_{-i}) 
  \le \hat g_i(\tilde v_i),
  \forall \tilde v_i\in\calD_i, \forall i\in\calN \label{long:eq:learn:ir}
\end{align}
where $\calD_i\equiv\{v_i\in\calV_i\mid \exists \tilde v\in\calD \mbox{ s.t. } \tilde v_i=v_i\}$ is the set of the types of player $i$ that appear in $\calD$.
More precisely, we use a lower confidence bound for $\hat g(\tilde v_i)$ in \eqref{long:eq:learn:ir} and an upper confidence bound for the right-hand side of \eqref{long:eq:learn:wbb}.
The learning problem \eqref{long:eq:learn:obj}-\eqref{long:eq:learn:ir} matches the optimization problem \eqref{long:eq:lp:obj}-\eqref{long:eq:lp:ir} in the limit of $|\calD|\to\infty$ when the law of large numbers applies to the sample averages (e.g., $\tilde v_i(\phi^\star(\tilde v))$ has finite variance), the regressors are asymptotically consistent (i.e., $\hat g_i(v_i)\to\E[v_j(\phi^\star(v))\mid v_i]$), the optimal $h$ is in the class of $h^\theta$ (i.e., realizable), the support of the prior distribution covers the whole $\calV$, and $\calV$ is finite.  

Although the use of $h^\theta$ alleviates Challenge \romannum{1}), the issue of computational complexity still remains.  Namely, each $h_i^{\theta_i}$ takes types $v_{-i}$ as its input, but each type is a function $v_j:2^\Omega\to\R$.  A question is how to represent those functions.  One may for example assume that the types are in a certain parametric family of functions (e.g., neural networks); then their parameters may be given as the input to $h_i^{\theta_{-i}}$ \cite{faccio2021parameterbased}.

We take an alternative approach of representing a function with its output values \cite{arxiv.2002.11833}.  Observe that $v_i:2^\Omega\to\R$ can be fully characterized by a $2^{|\Omega|}$-dimensional valuation-vector, $(v_i(\Phi))_{\Phi\subseteq\Omega}$, that represents the output values for all of the possible inputs.  Each $h_i: \calV_{-i}\to\R$ can then be seen as a function that maps $|\calN|-1$ vectors, each having $2^{|\Omega|}$ dimensions, to a real number; namely, $h_i:\R^{2^{|\Omega|}\,(|\calN|-1)}\to\R$ (see Figure~\ref{long:fig:h-wbb-ir}(a)).

\begin{figure}[t]
\begin{minipage}{0.49\linewidth}
  \centering
    \includegraphics[width=0.9\linewidth]{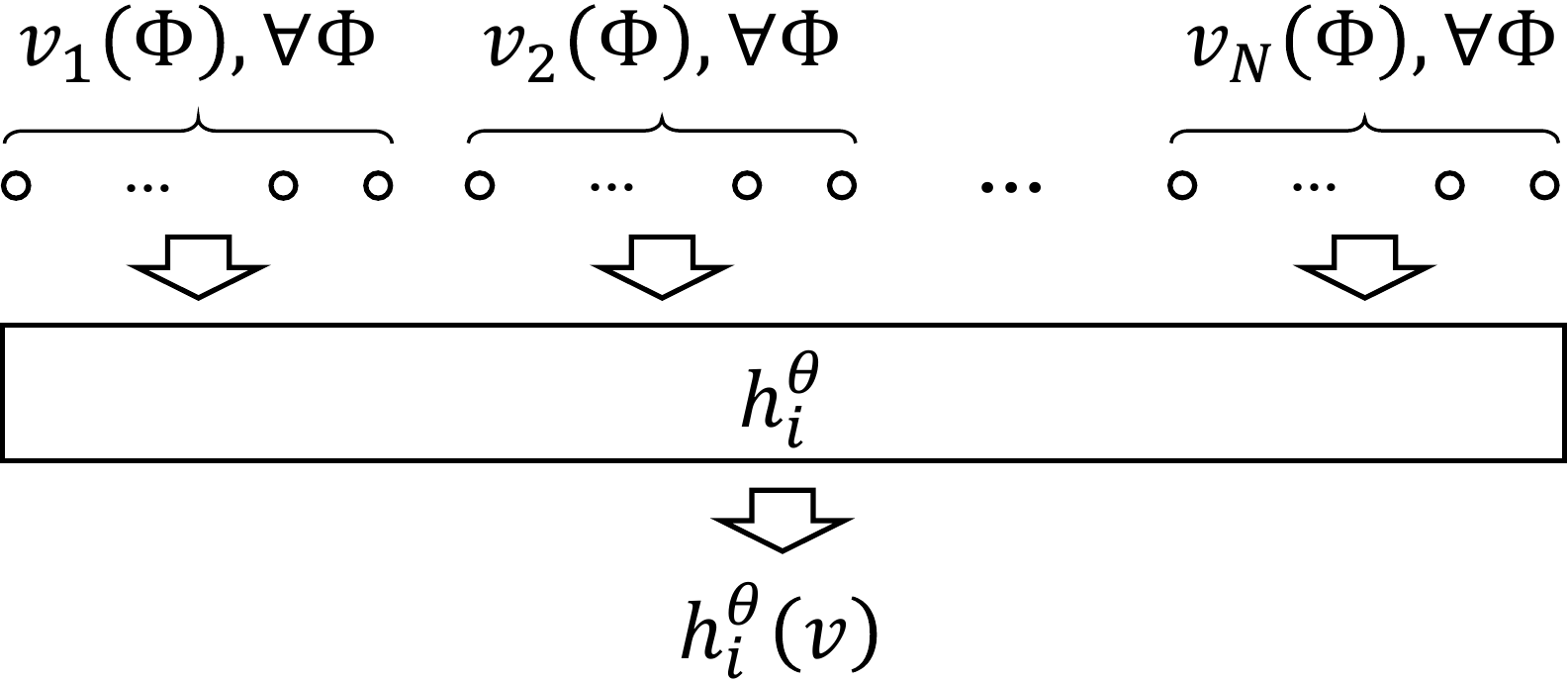}\\
    (a)
\end{minipage}
\begin{minipage}{0.49\linewidth}
    \centering
    \includegraphics[width=0.9\linewidth]{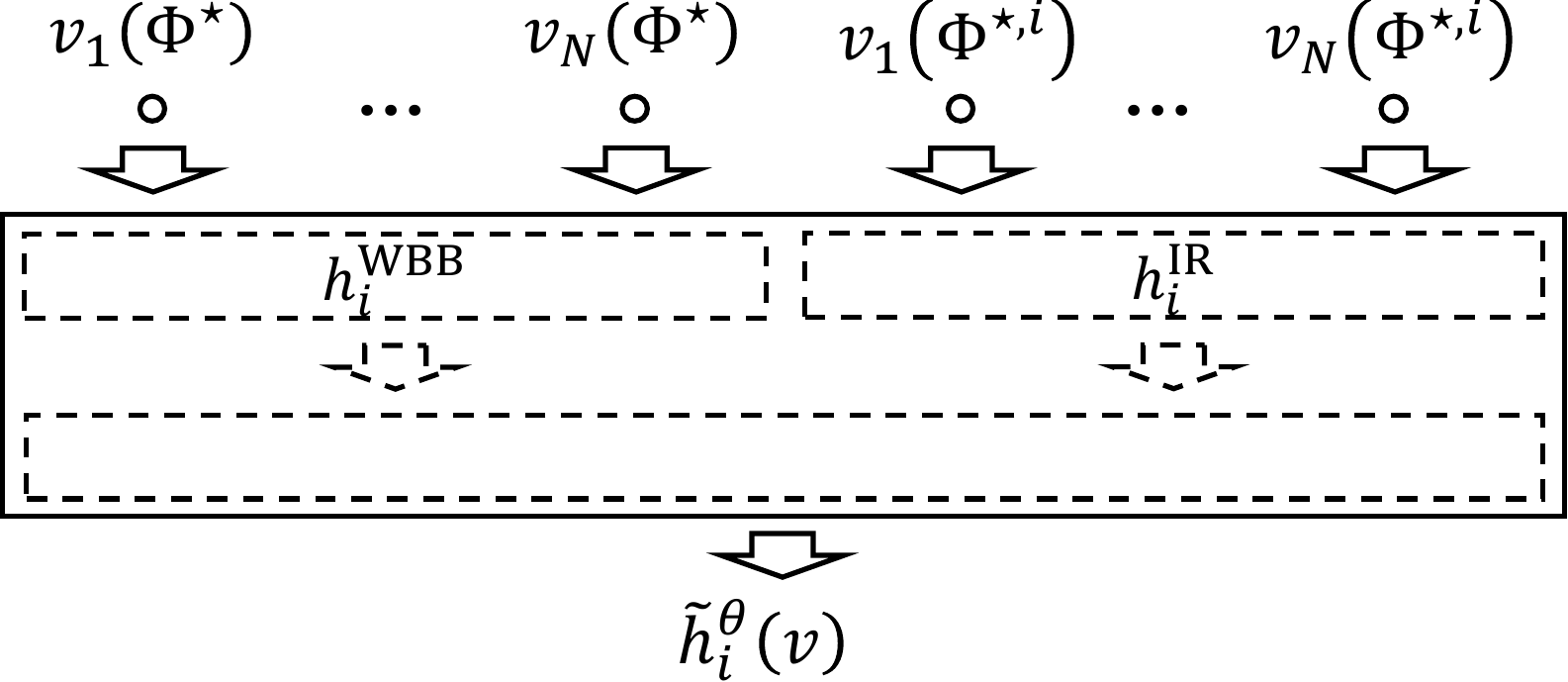}\\
    (b)
\end{minipage}
\caption{(a) A type $v_j$ represented by its output values, $v_j(\Phi)$ for $\Phi\subseteq\Omega$, is given as input to $h_i^\theta$.  (b) A function that takes $v_j(\Phi^\star)$ and $v_j(\Phi^{\star,i})$ for $j\neq i$ as input can arbitrarily combine $h_i^\WBB$ and $h_i^\IR$.}
\label{long:fig:h-wbb-ir}
\end{figure}


\subsection{Variable Reduction}
\label{long:sec:reduction}

The valuation-vector has exponentially many dimensions and does not fully resolve Challenge \romannum{1}).  We now show, based on game theoretic analysis, that some of the features of the valuation-vector are particularly important to achieve some of the desirable properties.  A low dimensional valuation-vector can then be formed with these important features.


To this end, we have designed the following special Groves mechanisms:
\begin{align}
  h_i^\WBB(v)
  & \equiv \max_{\Phi\subseteq\Omega} \sum_{j\neq i} v_j(\Phi), \forall i\in\calN \label{long:eq:groves:wbb}\\
  h_i^\IR(v)
  & \equiv \max_{\Phi\subseteq\Omega_{-i}} \sum_{j\neq i} v_j(\Phi), \forall i\in\calN, \label{long:eq:groves:ir}
\end{align}
where $\Omega_{-i} \equiv \Omega\setminus(\Omega_{i\to}\cup\Omega_{\to i})$ is the set of trades where player $i$ is neither the seller nor the buyer.  We will utilize these mechanisms, because it can be shown that $h^\WBB$ is always \expost{} WBB, and $h^\IR$ is \expost{} IR when no players reduce the social welfare, where a player reducing the social welfare is defined as follows:
\begin{definition}
  In a trading network $(\calN^+,\Omega,\calV)$, a player $i\in\calN$ with type $v_i\in\calV_i$ is called a negative player if there exists $v_{-i}\in\calV_{-i}$ such that
  \begin{align}
    \max_{\Phi\in\Omega_{-i}} \sum_{j\neq i} v_j(\Phi)
    & > \max_{\Phi\in\Omega} \sum_{j\in\calN} v_j(\Phi).
  \end{align}
\end{definition}
Formally,
\begin{theorem}
  For any trading network $(\calN^+,\Omega,\calV)$, the $h^\WBB$ in \eqref{long:eq:groves:wbb} satisfies DSIC, Efficiency, and \expost{} WBB, and the $h^\IR$ in \eqref{long:eq:groves:ir} satisfies DSIC, Efficiency, and \expost{} IR if $(\calN^+,\Omega,\calV)$ has no negative players.
\end{theorem}
\begin{proof}
  Since DSIC and Efficiency are guaranteed with the Groves mechanism, it suffices to prove \exante{} WBB and \interim{} IR, respectively.  Under $h^\WBB$, the payment from $i$ to IP is
  \begin{align}
    \tau_i^\WBB(v)
    & = h_i^\WBB(v) -  \sum_{j\neq i} v_j(\phi^\star(v)) \\
    & = \max_{\Phi\subseteq\Omega} \sum_{j\neq i} v_j(\Phi) - \sum_{j\neq i} v_j(\phi^\star(v)) \\
    & \ge 0
  \end{align}
  for any $(i,v)\in\calN\times\calV$.  We thus have $\sum_{i\in\calN} \tau_i^\WBB(v)\ge 0$ for any $v\in\calV$, which implies \expost{} WBB \eqref{long:eq:expostWBB}.  Also, when there are no negative players, the utility of player $i$ under $h^\IR$ is
  \begin{align}
    v_i(\phi^\star(v)) - \left(h_i^\IR(v) -  \sum_{j\neq i} v_j(\phi^\star(v))\right)
    & = \sum_{j\in\calN} v_j(\phi^\star(v)) - \max_{\Phi\subseteq\Omega_{-i}} \sum_{j\neq i} v_j(\Phi)\\
    & \ge 0
  \end{align}
  for any $v\in\calV$, which implies \expost{} IR \eqref{long:eq:expostIR}.
\end{proof}

Since $h_i^\WBB$ and $h_i^\IR$, respectively, achieve \expost{} WBB and  \interim{} IR, one might expect that their combination may simultaneously achieve the weaker properties of \exante{} WBB and
\interim{} IR.  Observe that an arbitrary combination of two functions can be represented by a function that takes the input of the two functions.  Since
\begin{align}
  h_i^\WBB(v)=\sum_{j\neq i} v_j(\Phi^\star),
  & \mbox{ where } \Phi^\star\in\argmax_{\Phi\subseteq\Omega} \sum_{j\neq i} v_j(\Phi) \notag\\
  h_i^\IR(v)=\sum_{j\neq i} v_j(\Phi^{\star,i}),
  & \mbox{ where } \Phi^{\star,i}\in\argmax_{\Phi\subseteq\Omega_{-i}} \sum_{j\neq i} v_j(\Phi),\notag
\end{align}
we see that $(v_j(\Phi^\star))_{j\neq i}$ and $(v_j(\Phi^{\star,i}))_{j\neq i}$ are the input of $h_i^\WBB$ and $h_i^\IR$, respectively, and hence are the important features of the valuation-vector.  We may thus seek to find the optimal function within the class of functions that take those important features as input (see Figure~\ref{long:fig:h-wbb-ir}).  We may also consider a larger class of functions by allowing additional (e.g., random) elements of the original valuation-vector as input, trading off the quality of approximation against computational complexity.

\section{Experiments}
\label{long:sec:exp}

We conduct our experiments with trading networks where it is impossible to achieve all of the four desirable properties \emph{\expost{}}.  Specifically, we generate non-trivial instances of trading networks uniformly at random in the setting of Example~\ref{long:ex:two}, which involves a single potential trade between a seller and a buyer\footnote{
To generate non-trivial instances, we first randomly generate 10,000
instances that are not necessarily non-trivial (for each player, pairs
of numbers in [0,1] are generated uniformly at random and let those
values be the cost of the two types; the prior probabilities of the
types are also chosen uniformly at random).  Out of those instances,
1,642 have non-trivial $\calV$ and used in the experiments.
}.
By Theorem~\ref{long:thrm:impos}, there exist no mechanisms that satisfy all of the four desirable properties \emph{\expost{}} for those non-trivial instances.  We apply our proposed AMD approaches to these trading networks and study whether they can find the mechanisms that satisfy the four properties if WBB is \exante{} and IR is \interim{}.

\subsection{Validating Computational Approach}
\label{long:sec:exp:comp}

We first validate our computational approach by applying it to each of the randomly generated instances.  Here, we solve the optimization problem in \eqref{long:eq:lp:obj}-\eqref{long:eq:lp:ir} via CPLEX 22.1.

\begin{figure}[t]
  \begin{minipage}{0.49\linewidth}
    \centering
    \includegraphics[width=\linewidth]{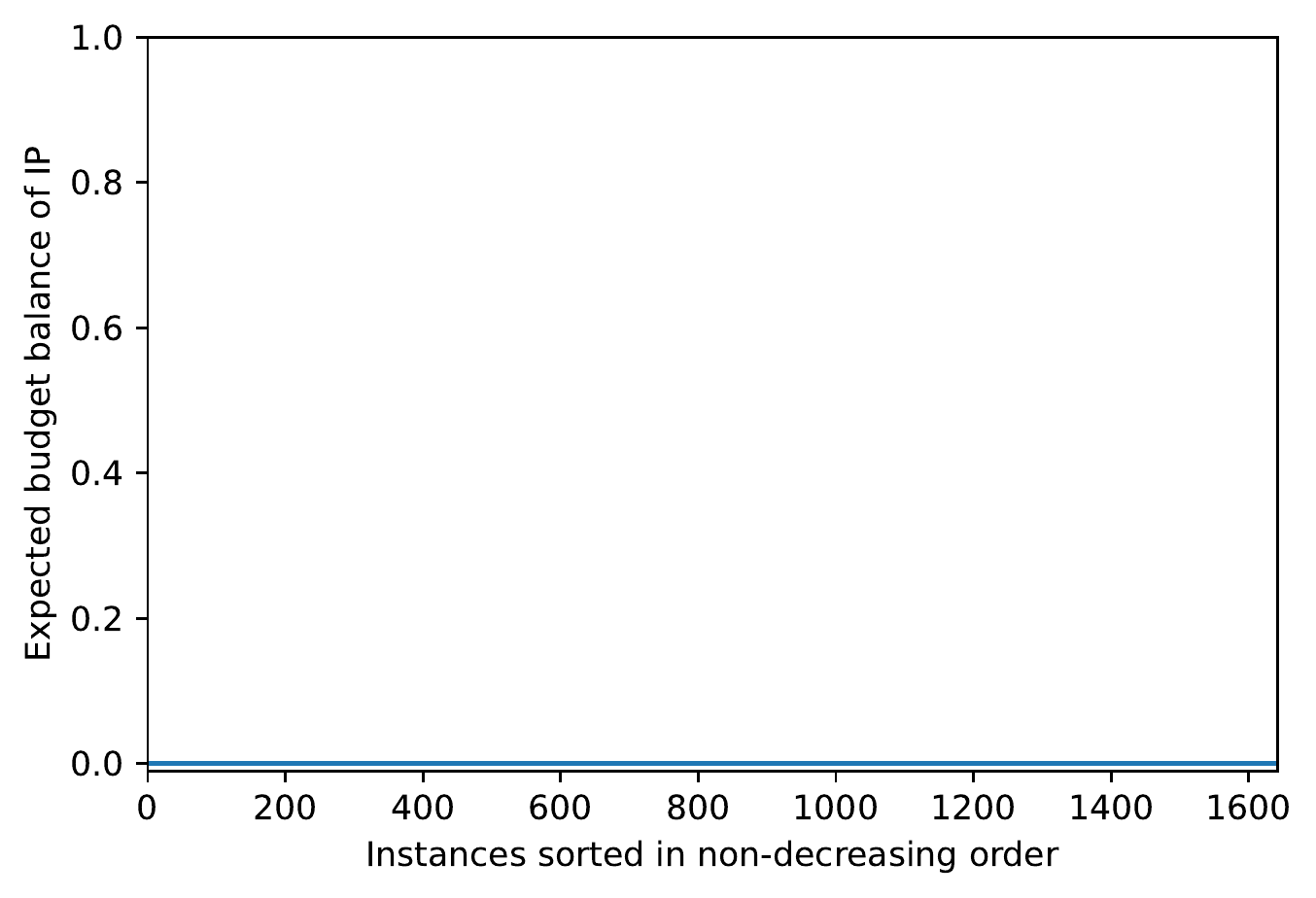}\\
    (a) Budget balance
  \end{minipage}
  \begin{minipage}{0.49\linewidth}
    \centering
    \includegraphics[width=\linewidth]{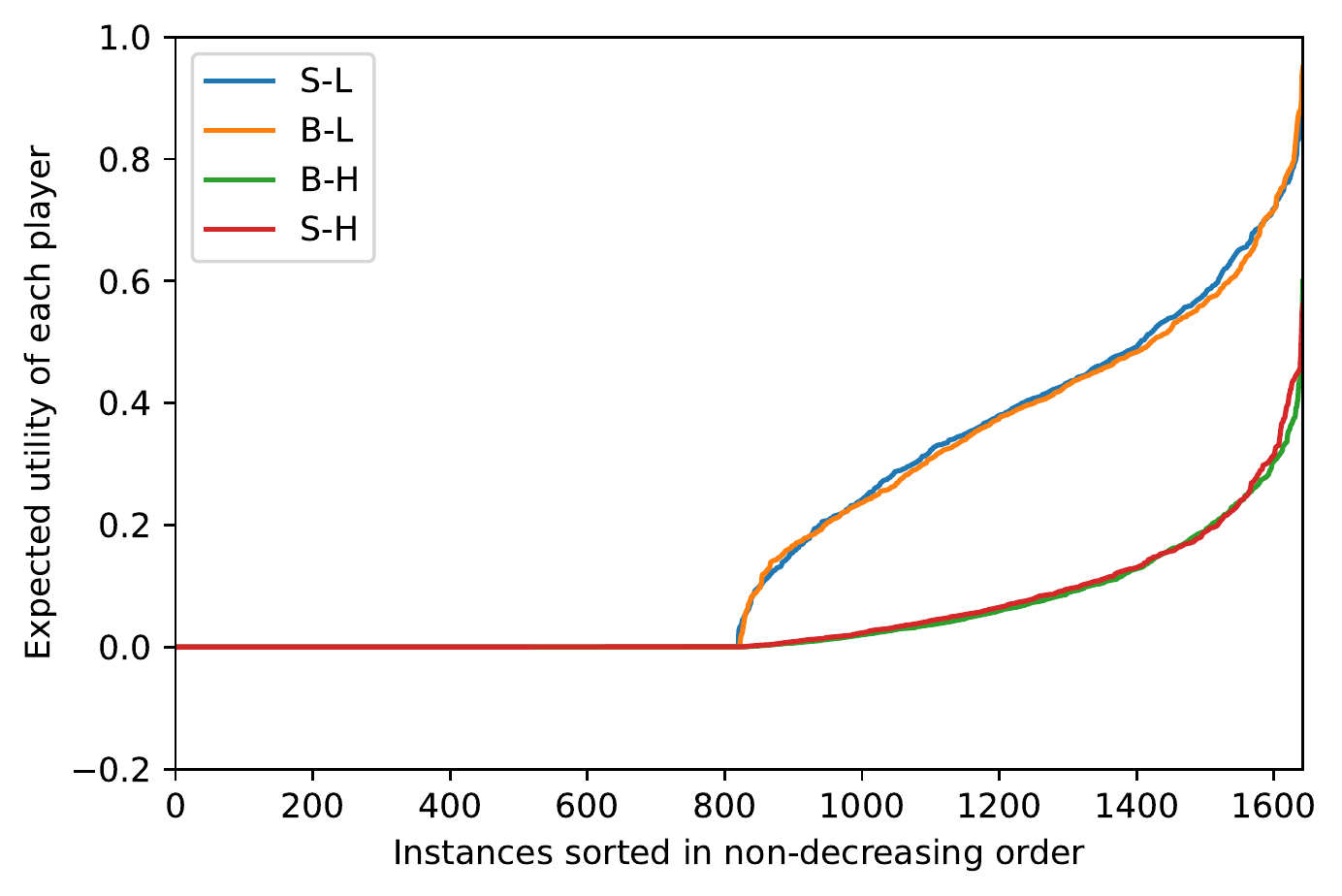}\\
    (b) Utility
  \end{minipage}
  \caption{The performance of the mechanisms computed by the proposed approach, plotted in non-decreasing order.  Panel (a) shows the expected budget balance of IP, and panel (b) shows expected utility of each player (S-L denotes the seller of low cost, and B-H denotes the buyer of high costs; S-H and B-L are defined analogously).}
  \label{long:fig:comp}
\end{figure}

Figure~\ref{long:fig:comp} shows the expected budget balance (utility)
of IP in (a) and the expected utility of each player in (b) for the
1,642 instances plotted in non-decreasing order.  Observe that, for
all instances, IP expects essentially zero budget balance, and
every player (supplier or buyer of any type) expects non-negative utility.  Therefore,
the proposed approach finds the mechanism that achieves both of
\exante{} WBB\footnote{In fact, strong budge balance is achieved.}  and
\interim{} IR (in addition to DSIC and Efficiency) for every randomly
generated instance where achieving the four properties \emph{\expost{}} is impossible. 


This however does not mean that one can \emph{always} achieve \exante{} WBB and
\interim{} IR.  There indeed exist instances for which \exante{} WBB or
\interim{} IR cannot be achieved.  Empirically, however, such infeasible
instances appear to have zero Lebesgue measure and are not generated
in our random process.  A caveat is that, when an instance is close to
an infeasible instance, our approach tends to find the
mechanism that requires a large amount of payment to or from IP.  The
amount of payment occasionally exceeds the cost of each player by more
than 100 times, which is unlikely to be acceptable in practice even if
the expected utility is non-negative.  See the expected payment of each
player to/from IP in Figure~\ref{long:fig:comp-payment}.

\begin{figure}[t]
  \centering
  \begin{minipage}{0.49\linewidth}
    \centering
    \includegraphics[width=\linewidth]{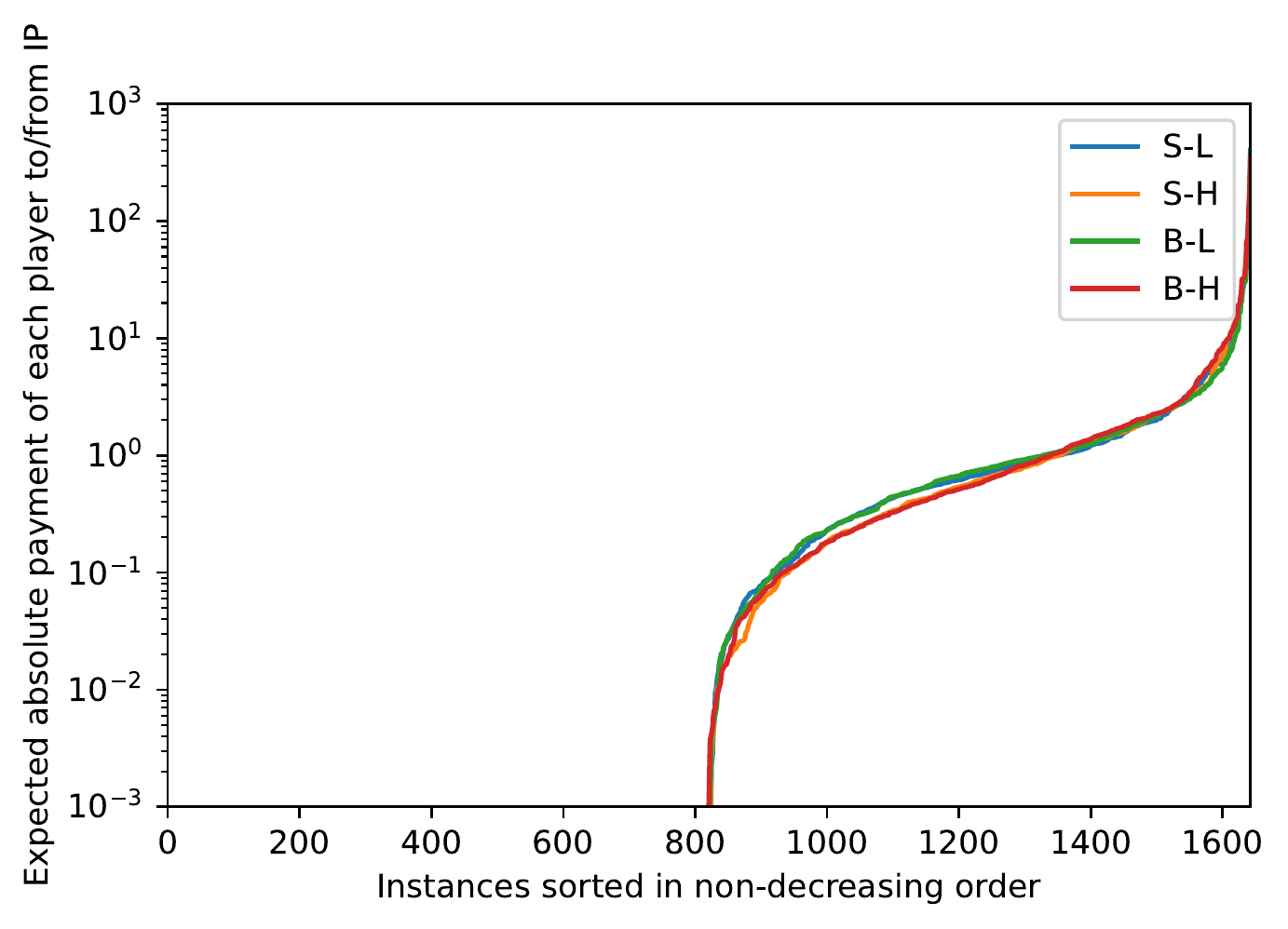}
  \end{minipage}
  \caption{The expected absolute payment of each player to/from IP, plotted in non-decreasing order.}
  \label{long:fig:comp-payment}
\end{figure}

\subsection{Validating Learning Approach}

Next, we validate the effectiveness of our learning approach with a focus on the impact of sample approximation.  Here, for each of the 1,642 instances, we take varying number of samples from the prior distribution of types to solve the learning problem \eqref{long:eq:learn:obj}-\eqref{long:eq:learn:ir}.  Here, we use the lower confidence bound of a Gaussian process regressor to evaluate the right-hand side of \eqref{long:eq:learn:ir} and added the sample standard deviation to the right-hand side of \eqref{long:eq:learn:wbb} to form an upper confidence bound.  We do not rely on functional approximation $h^\theta$ and seek to find an exact $h$.

\begin{figure}[t]
  \centering
  \begin{minipage}{0.49\linewidth}
    \centering
    \includegraphics[width=\linewidth]{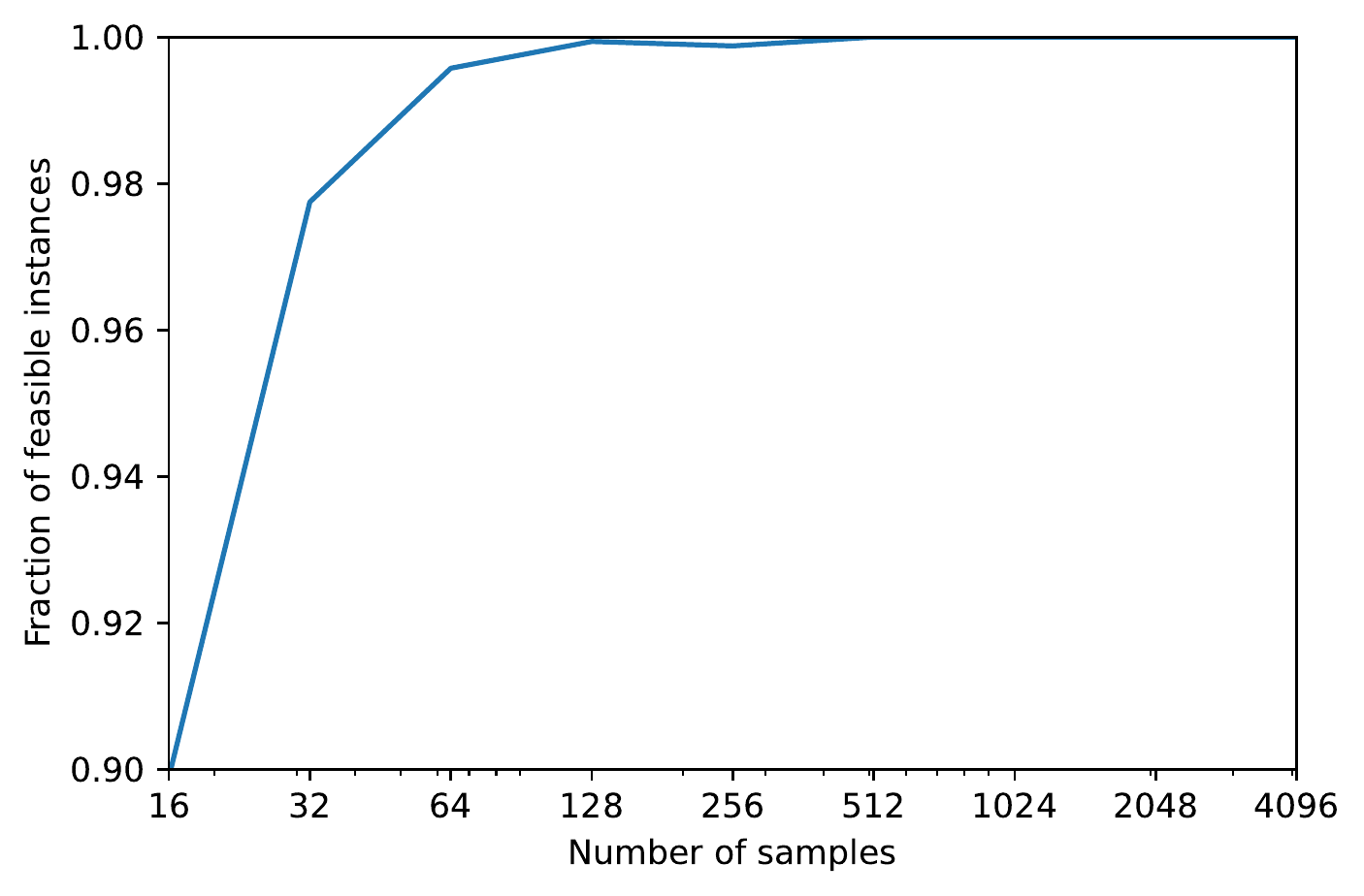}\\
    (a) Feasibility 
  \end{minipage}
  \begin{minipage}{0.49\linewidth}
    \centering
    \includegraphics[width=\linewidth]{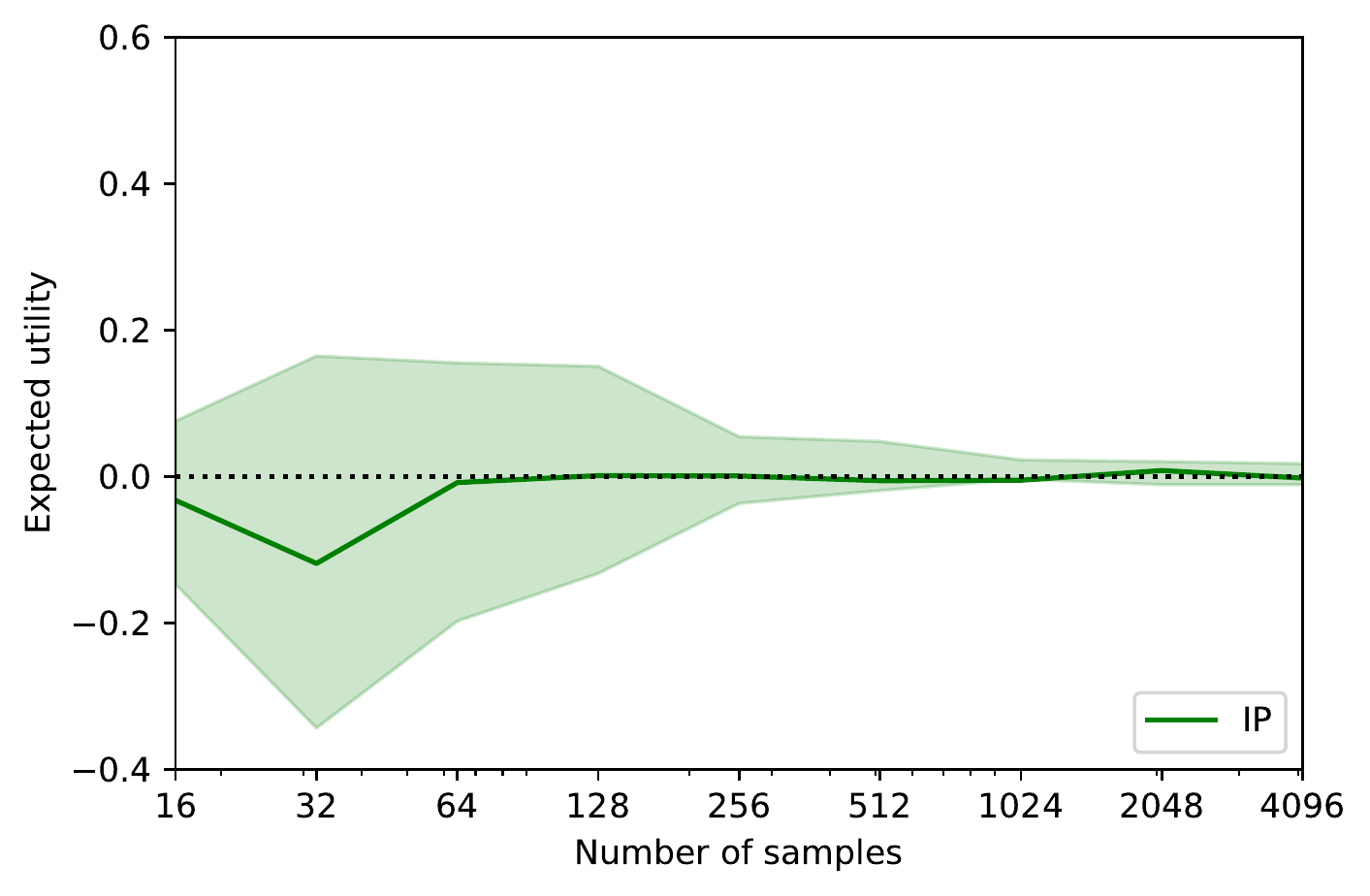}\\
    (b) Budget balance
  \end{minipage}
  \begin{minipage}{0.49\linewidth}
    \centering
    \includegraphics[width=\linewidth]{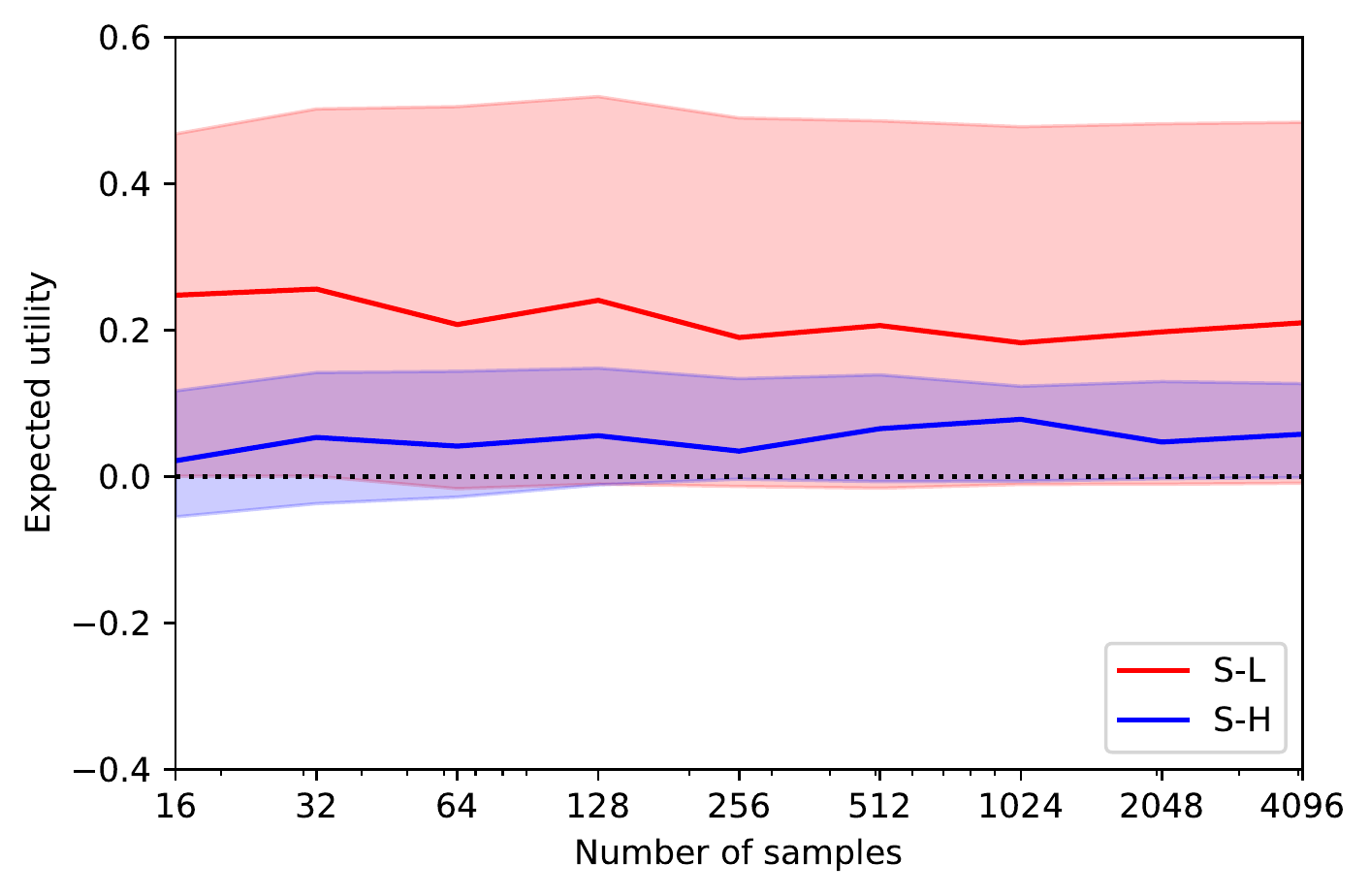}\\
    (c) Supplier utility
  \end{minipage}
  \begin{minipage}{0.49\linewidth}
    \centering
    \includegraphics[width=\linewidth]{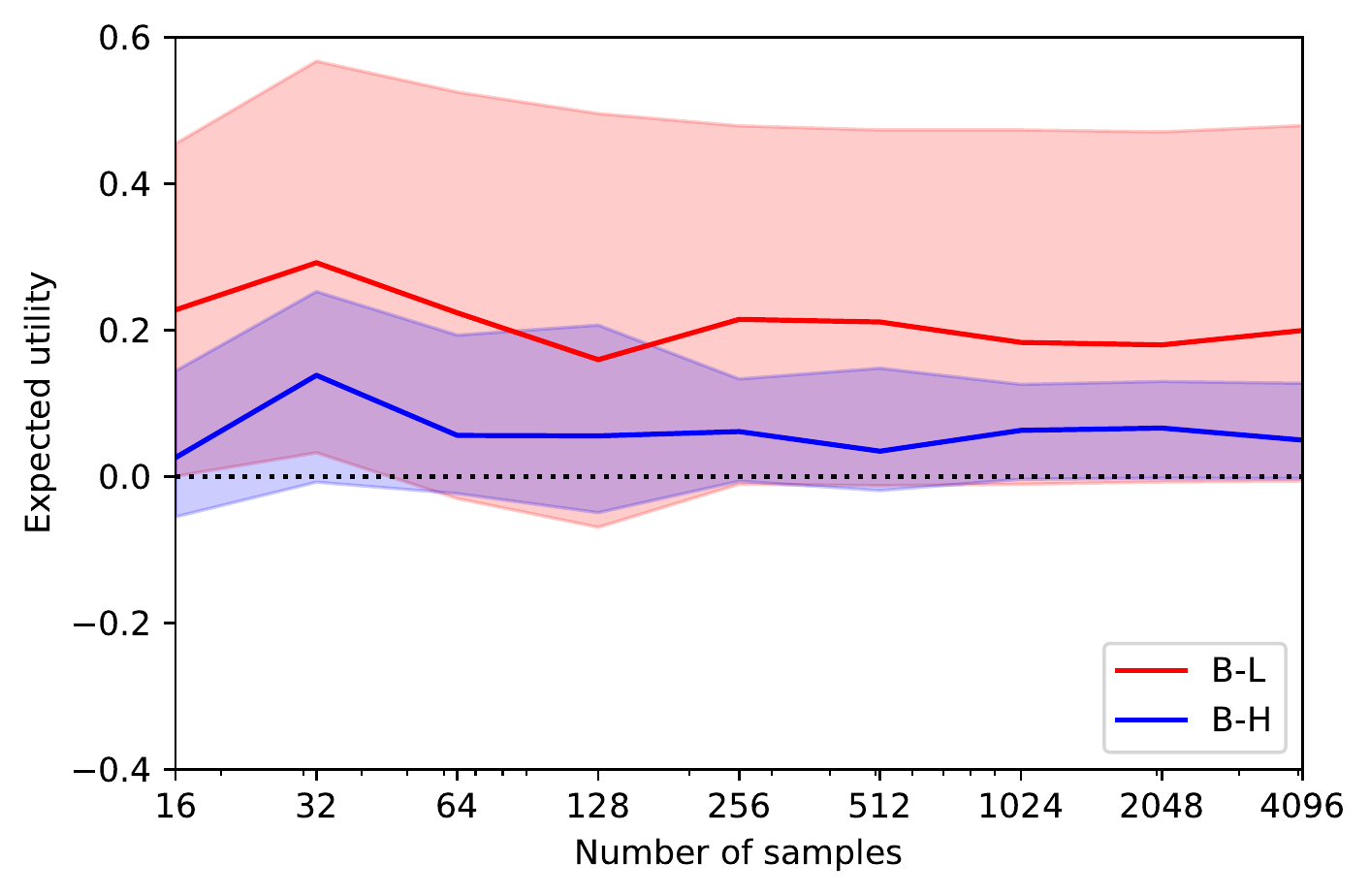}\\
    (d) Buyer utility
  \end{minipage}
  \caption{The performance of the mechanisms learned by our
    approach against the number of the samples from the prior
    distribution.  Panel (a) shows the fraction of the instances for
    which the learning problem has feasible solutions.  For those
    feasible instances, panel (b) shows the expected budget balance of
    IP, and panel (c)-(d) shows expected utility of each player.
    The solid line is the average expected utility, and the shaded area is the [15.9\%, 84.1\%]-quantile.}
  \label{long:fig:learn}
\end{figure}

Figure~\ref{long:fig:learn} shows (a) the fraction of the instances for which the learning problem has feasible solutions as well as the expected utility of (b) IP, (c) the supplier, and (d) the buyer under the mechanisms learned with our approach.  In (b)-(d), the solid lines show the average over the feasible instances, and the shared area represents the [15.9\%, 84.1\%]-quantile, which coincide with the confidence intervals set in \eqref{long:eq:learn:wbb}-\eqref{long:eq:learn:ir}.

We find that the learning problems do not necessarily have feasible
solutions (while the corresponding optimization problems are all
feasible), but with sufficient amount of data ($\ge 512$ in this
case), all of the learning problems become feasible.  For those
feasible learning problems, the proposed approach learns the
mechanisms where the supplier and the buyer has strictly positive
expected utility (achieving \interim{} IR), and the IP has approximately
zero expected utility (achieving \exante{} WBB).  Also, \exante{} WBB and
\interim{} IR are achieved with high probability, but they are sometimes
violated unlike the mechanisms computed with the exact knowledge of
the prior distribution.  As the sample size increases, however,
\exante{} WBB and \interim{} IR get more frequently satisfied with the
mechanisms found with learning.

\subsection{Validating Variable Reduction}
\label{long:sec:exp:reduce}


Finally, we study the effectiveness of variable reduction by applying it to the settings of the previous experiments.
Recall that the number of
variables can be reduced from $2^{|\Omega|} \, |\calN| \, (|\calN| -
1)$ to $2 \, |\calN| \, (|\calN| - 1)$.  While this may not appear to
reduce the number of variables when $|\Omega|=1$ and $|\calN|=2$, it
reduces the essential number of variables even in this case.
Specifically, for the trading network in Example~\ref{long:ex:two}, we have
\begin{align}
  \argmax_{\Phi\subseteq\Omega} v_\rmS(\Phi)
  = \argmax_{\Phi\subseteq\Omega_{-\rmB}} v_\rmS(\Phi)
  = \{\emptyset\},
\end{align}
so that both $h_\rmB^\IR$ and $h_\rmB^\WBB$ are functions $v_\rmS(\emptyset)$.


Figure~\ref{long:fig:comp_reduced} shows the performance of the mechanisms
computed by the proposed approach with the technique of variable
reduction.  Compared to the corresponding results in
Figure~\ref{long:fig:comp} and Figure~\ref{long:fig:comp-payment} without variable reduction,
we find different mechanisms depending on whether the variables are reduced or not, as is suggested by the different expected utilities of the players.  However, the mechanisms computed with reduced variables still achieve \exante{} WBB and \interim{} IR for all of the randomly generated instances.

\begin{figure}[t]
  \begin{minipage}{0.33\linewidth}
    \centering
    \includegraphics[width=\linewidth]{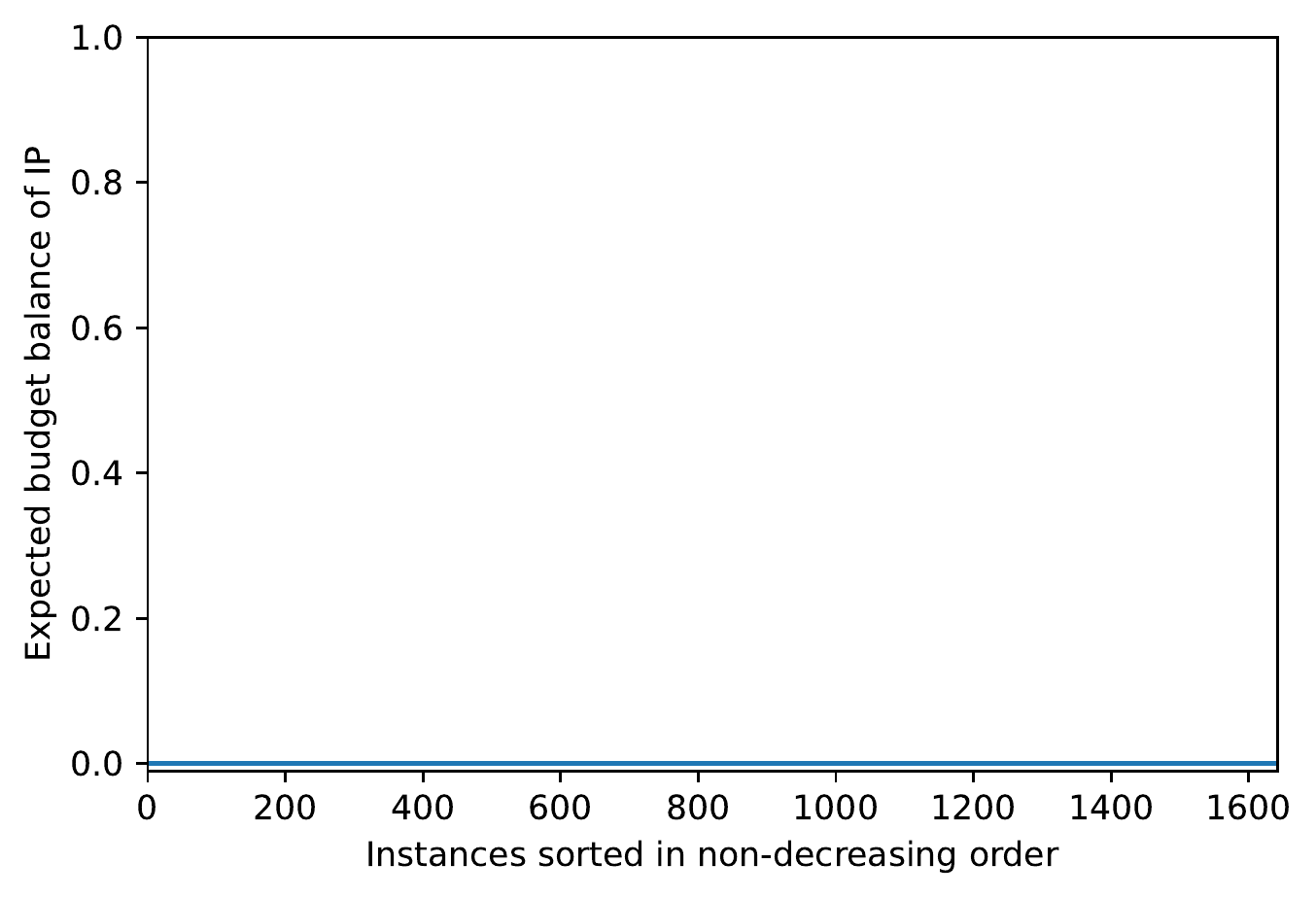}\\
    (a) Budget balance
  \end{minipage}
  \begin{minipage}{0.33\linewidth}
    \centering
    \includegraphics[width=\linewidth]{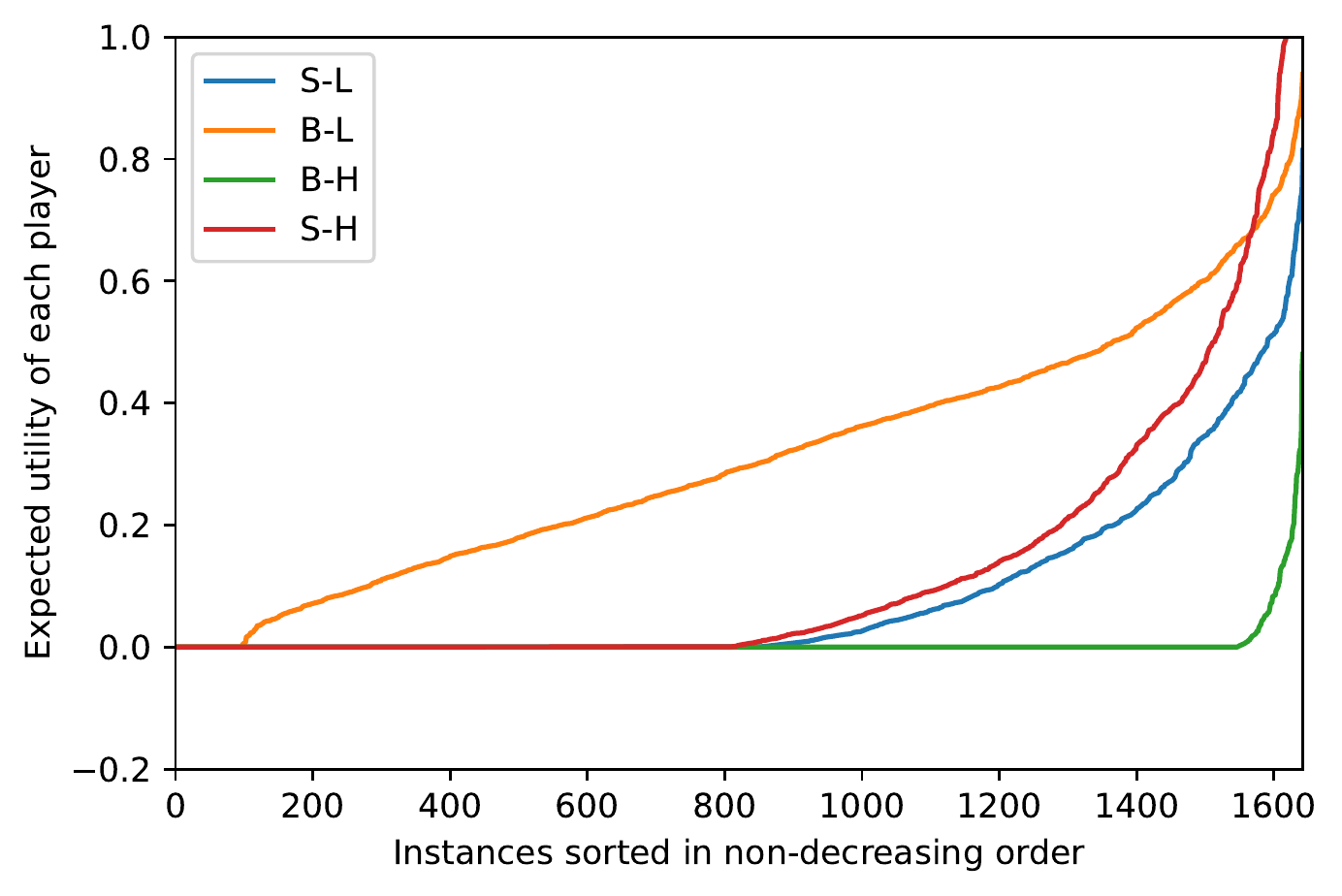}\\
    (b) Utility
  \end{minipage}
  \begin{minipage}{0.33\linewidth}
    \centering
    \includegraphics[width=\linewidth]{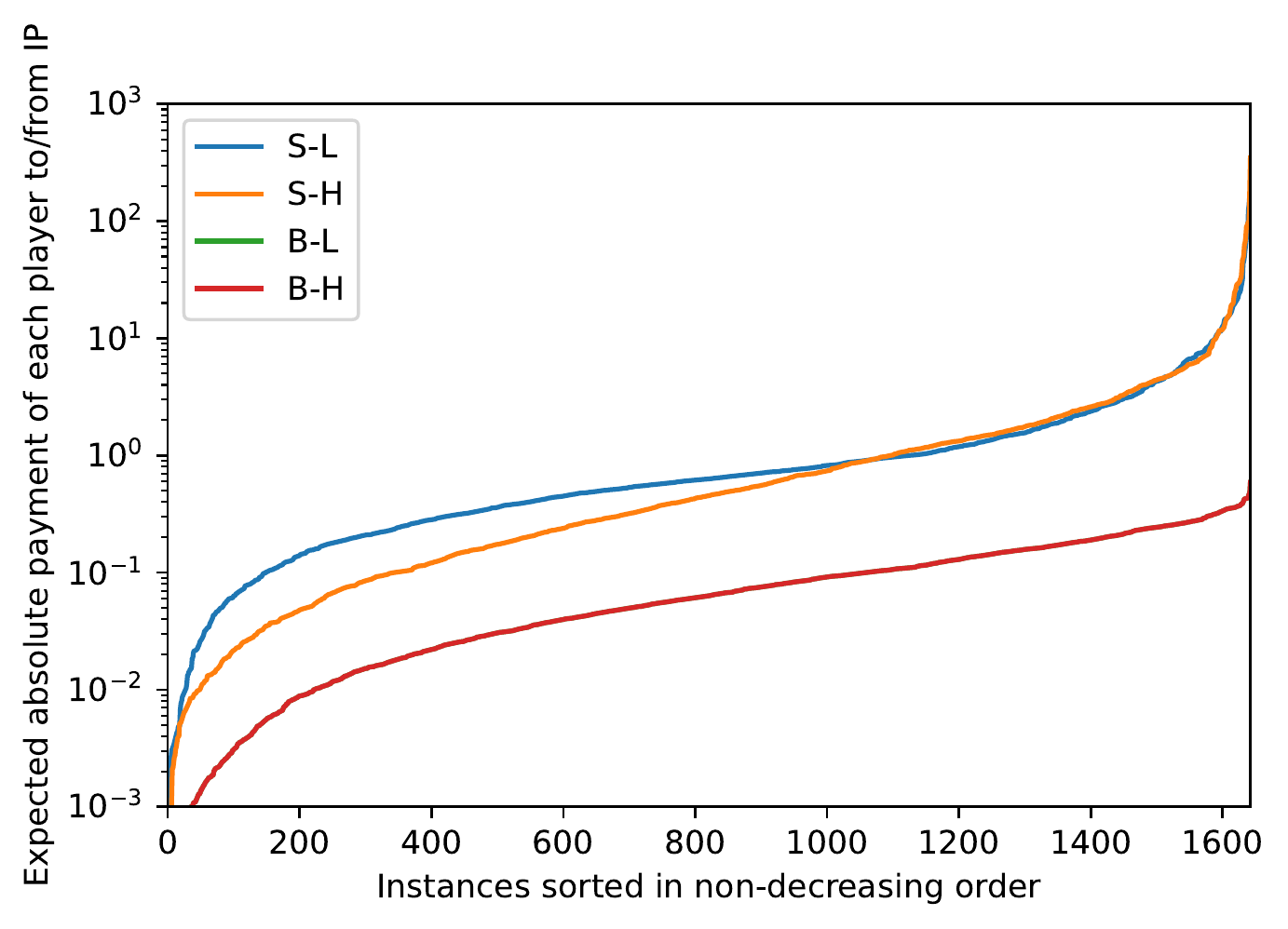}\\
    (c) Payment to/from IP
  \end{minipage}
  \caption{The performance of the mechanisms computed by the proposed approach with reduced variables.}
  \label{long:fig:comp_reduced}
\end{figure}

Figure~\ref{long:fig:learn_reduced} shows the performance of the mechanisms \emph{learned} with variance reduction.
Similar to the corresponding results without variable reduction 
(Figure~\ref{long:fig:learn}), \exante{} WBB and \interim{} IR are achieved with
high probability when the learning problems are feasible.
Variable reduction, however, reduces the solution space, and the instances that are feasible without variable reduction can become infeasible.  Nevertheless, the fraction of feasible instances increases with the sample size (and all of the instances become feasible with the exact prior distribution, as we have shown with the computational approach).

\begin{figure}[t]
  \begin{minipage}{0.49\linewidth}
    \centering
    \includegraphics[width=\linewidth]{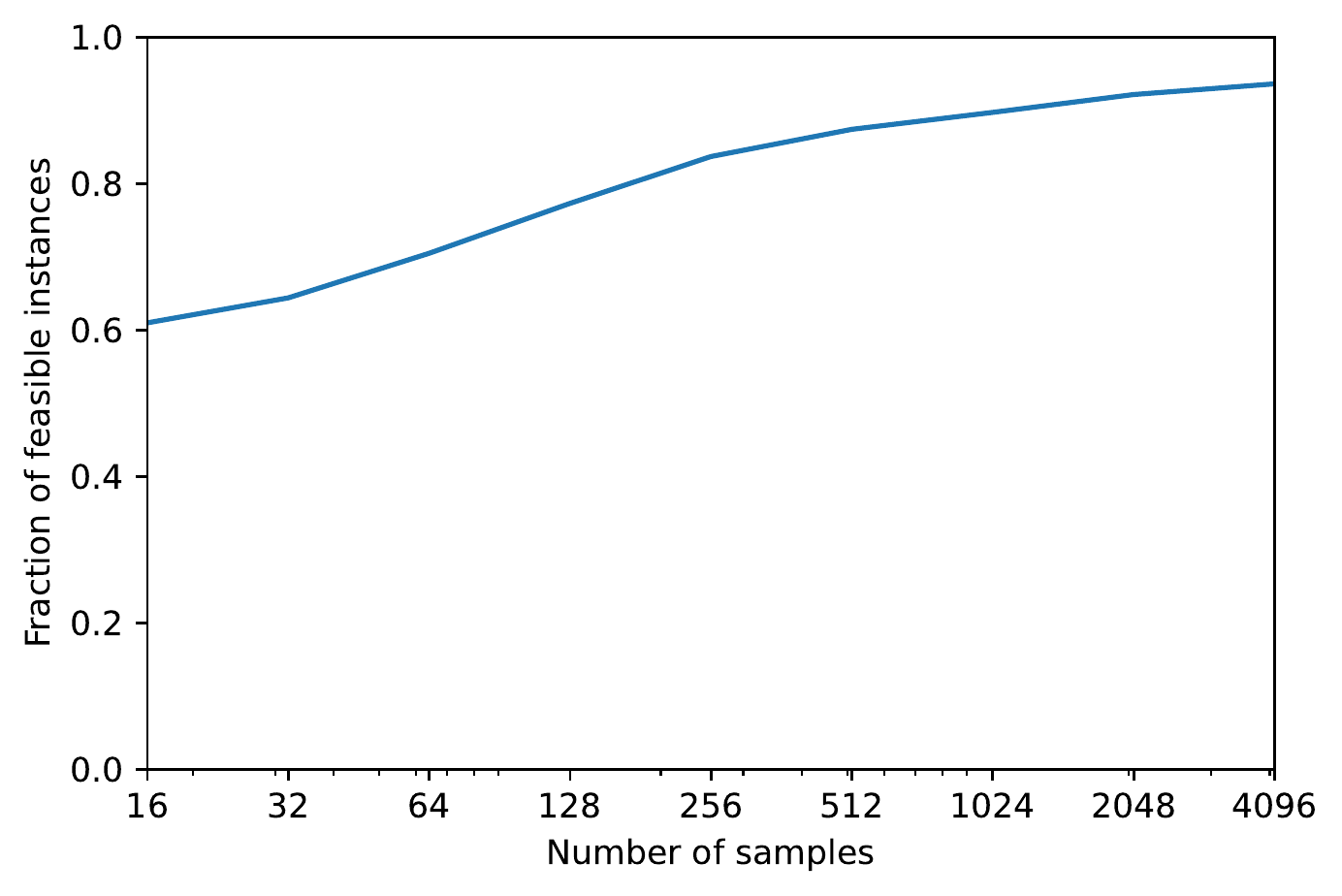}\\
    (a) Feasibility 
  \end{minipage}
  \begin{minipage}{0.49\linewidth}
    \centering
    \includegraphics[width=\linewidth]{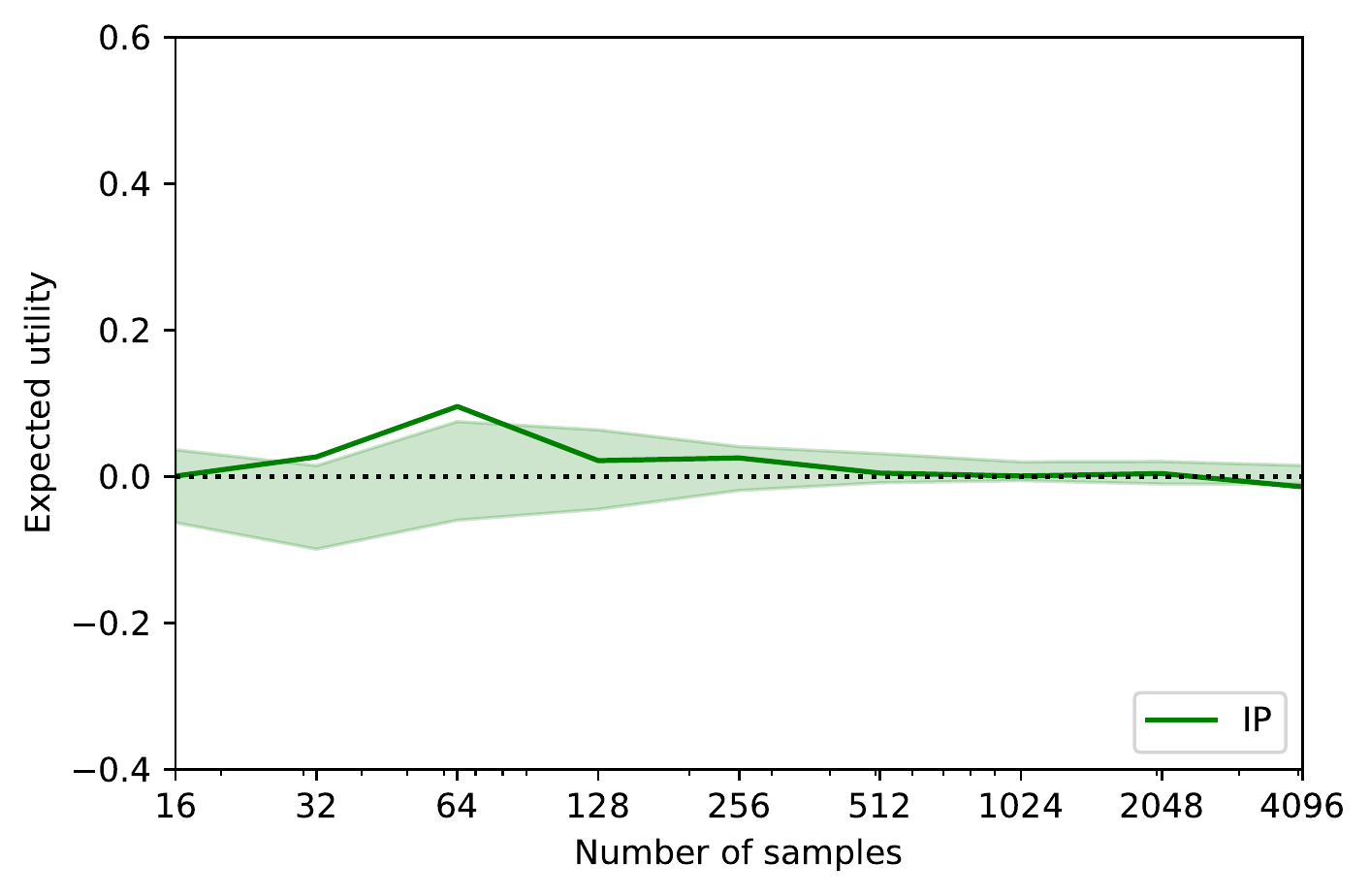}\\
    (b) Budget balance
  \end{minipage}
  \begin{minipage}{0.49\linewidth}
    \centering
    \includegraphics[width=\linewidth]{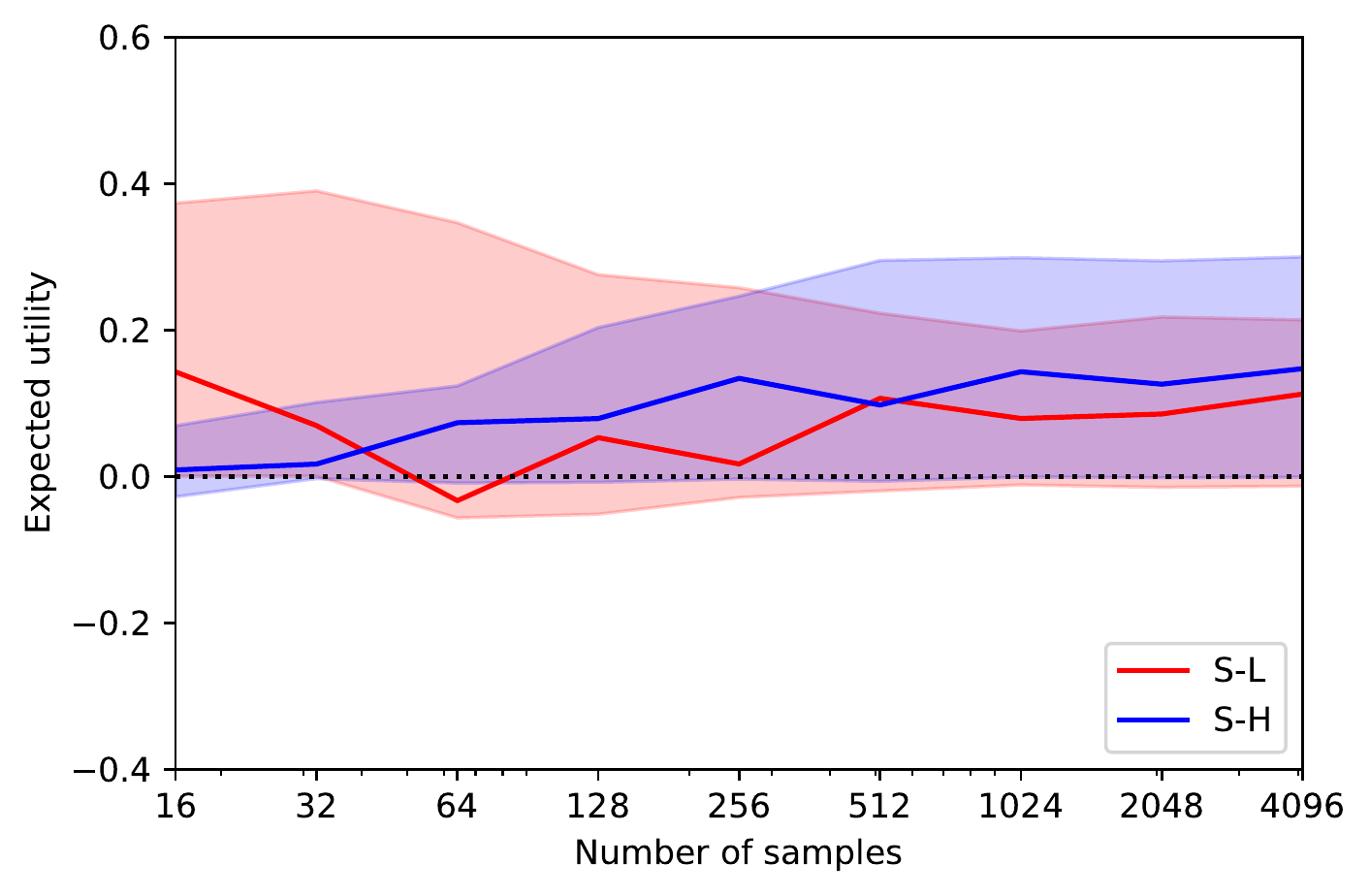}\\
    (c) Supplier utility
  \end{minipage}
  \begin{minipage}{0.49\linewidth}
    \centering
    \includegraphics[width=\linewidth]{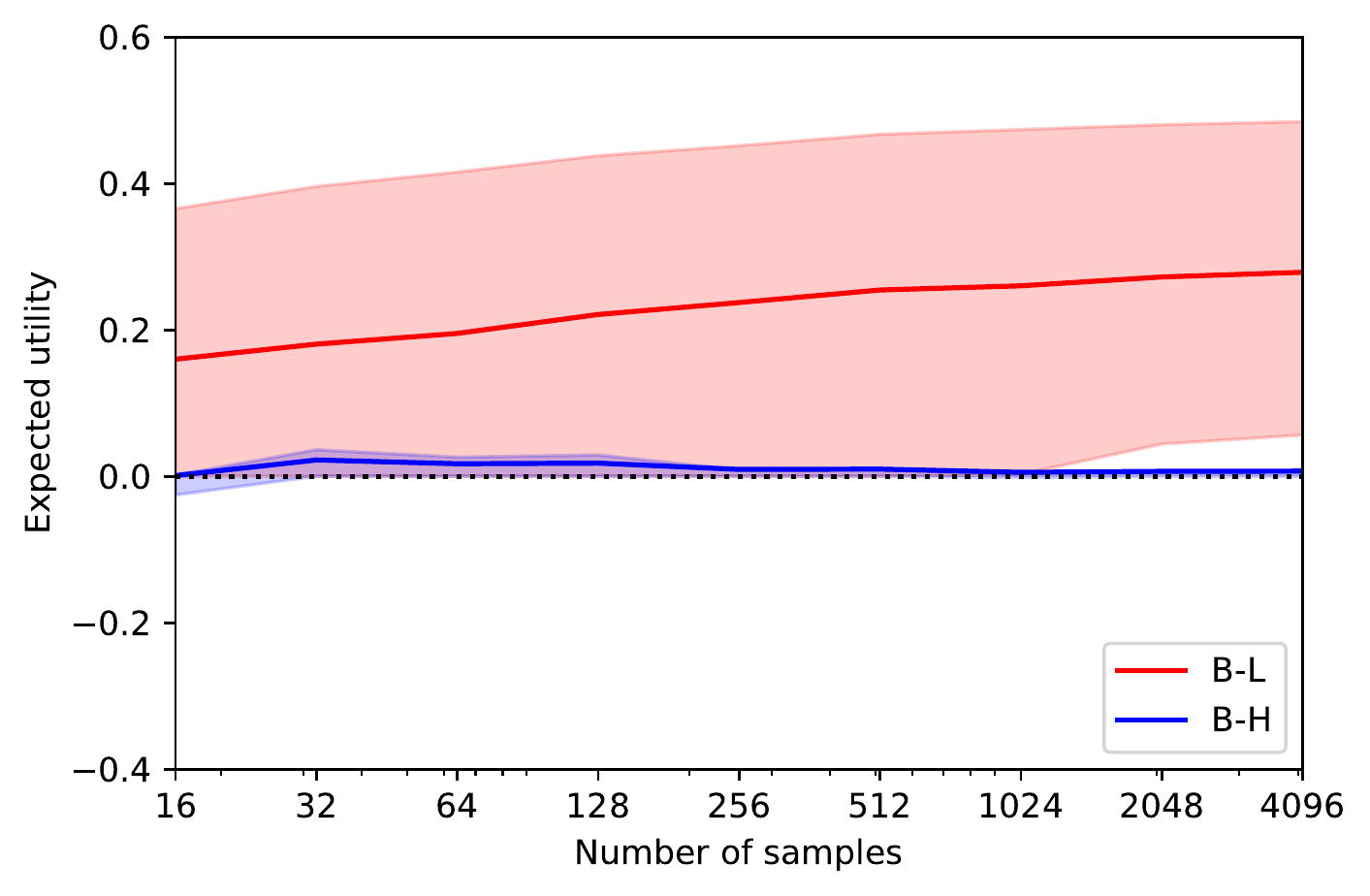}\\
    (d) Buyer utility
  \end{minipage}
  \caption{The performance of the mechanisms learned by the proposed
    approach with reduced variables.}
  \label{long:fig:learn_reduced}
\end{figure}

\section{Conclusion}
\label{long:sec:conclusion}

We have shown that, by computing or learning appropriate mechanisms, it can be made possible to achieve DSIC, Efficiency, \exante{} WBB, and \interim{} IR in trading networks where achieving these four properties \expost{} is impossible.  Our AMD approach is greatly simplified with Theorem~\ref{long:thrm:nopayment}, which allows us to ignore the payment between players, and the Groves mechanisms, which allows us to focus on achieving WBB and IR.

This paper proposes the first AMD approach for trading networks and suggests interesting directions of future work.  For example, while we learn mechanisms from previously collected (offline) data, the prior work has also investigated the approaches of learning mechanisms while collecting data in an online manner \cite{pmlr-v32-mohri14,CRJ08,balaguer2022hcmdzero,Liu_Chen_Qin_2015} or by reinforcement learning \cite{ijcai2017-RLAMDAdAuctions,IndirectMechanismDesignThesis,balaguer2022the,AIEconomist,abs-1806-04067,DemocraticAI,Brero_Eden_Gerstgrasser_Parkes_Rheingans-Yoo_2021,Bre21}.  Such online methods involves the additional challenge of the tradeoff between exploration and exploitation, and it is an interesting direction of future work to study such online methods for trading networks.
\cite{learningauctions-morgenstern16,MultiItemRevenueMaximizingFOCS2017,RevenueMaximizingAMDFocs2018,1530752} study sample complexity of learning mechanisms in combinatorial auctions, but such sample complexity for trading networks is unknown to date.

\SWdone{We may want to mention the following as future work:
\begin{itemize}
    \item Ex-post WBB.
    \item Theoretical bounds using Rademard complexity or pseudo dimensions.
    \item Learning indirect mechanisms.
    \item VCG redistribution to get closer to strong budget balance/stable outcome/competitive equilibrium
\end{itemize}
}

\bibliographystyle{aaai23}
\bibliography{game}

\clearpage
\appendix

\section{When the payment rule is given by other means}

Our approach may also be applied when $\pi$ is determined by other means.
Namely, for any payment rule $\pi$ for
$\calT^+(\calV)=(\calN,\Omega,\calV)$, we may design a mechanism
$(\phi,\tau',\pi_0)$ with no payment between players, and convert $\tau'$ to
$\tau$ via \eqref{long:eq:payment_conversion} to obtain the mechanism
$(\phi,\tau,\pi)$.  See Algorithm~\ref{long:alg:design}.

\begin{algorithm}
    \caption{Mechanism design for trading network with IP}
    \hspace*{\algorithmicindent} \textbf{Input: } A trading network $(\calN,\Omega,\calV)$ and a payment rule $\pi$; Desirable properties\\
    \hspace*{\algorithmicindent} \textbf{Output: } A mechanism $(\phi,\tau,\pi)$ that achieves the desirable properties
    \begin{algorithmic}[1]
    \STATE For $(\calN,\Omega,\calV)$, design a mechanism $(\phi,\tau',\pi_0)$ with no payment between players that achieves the desirable properties
    \STATE Let
    $\tau_i(v)
    =
    \tau_i'(v)
    + \sum_{\omega\in \phi(v)_{i \to}} \pi(\omega;v)
    - \sum_{\omega\in \phi(v)_{\to i}} \pi(\omega;v)$
    for any $(i,v)\in\calN\times\calV$
    \end{algorithmic}
    \label{long:alg:design}
\end{algorithm}

In particular, for an arbitrary payment rule $\pi$, one may first find a Groves mechanism $(\phi^\star,\tau')$  and transform it into a mechanism $(\phi^\star,\tau)$ by following the procedure in Algorithm~\ref{long:alg:design}.  Even though $(\phi^\star,\tau)$ may not be a Groves mechanism in the sense of \eqref{long:eq:groves_allocation}-\eqref{long:eq:groves_payment}, DSIC and Efficiency of $(\phi^\star,\tau)$ are guaranteed by the corresponding properties of $(\phi^\star,\tau')$ according to Theorem~\ref{long:thrm:nopayment}.

\end{document}